\documentclass[letterpaper, 11pt]{article}

\usepackage{amsthm}
\usepackage{verbatim} 
\usepackage{graphicx}
\usepackage{amsmath}
\usepackage{amssymb}
\usepackage{multirow}
\usepackage{array}
\usepackage{picinpar}
\usepackage{rotating}
\usepackage{fancybox}
\usepackage{wrapfig}
\usepackage[all,2cell]{xy}
\UseAllTwocells

\newcolumntype{P}{p{6mm}<{\centering} }
\newcolumntype{M}{m{2in}<{\centering} }
\newcolumntype{N}{m{1.55in}<{\centering} }

\theoremstyle{plain}
\newtheorem{thm}{Theorem}[section] 
\newtheorem{cor}[thm]{Corollary}
\newtheorem{obs}[thm]{Observation}

\theoremstyle{remark}
\newtheorem{exa}{Example}[section]

\theoremstyle{definition}
\newtheorem{define}{Definition}[section]

\begin{document}

\title{Measures of Tipping Points, Robustness, \\ and Path Dependence}
\author{Aaron L Bramson}
\date{ }
\maketitle

\vspace{.5in}

\begin{abstract}
This paper draws distinctions among various concepts related to tipping points, robustness, path dependence, and other properties of system dynamics.  For each concept a formal definition is provided that utilizes Markov model representations of systems.  We start with the basic features of Markov models and definitions of the foundational concepts of system dynamics.  Then various tipping point-related concepts are described, defined, and illustrated with a simplified graphical example in the form of a stylized state transition diagram.  The tipping point definitions are then used as a springboard to describe, formally define, and illustrate many distinct concepts collectively referred to as ``robustness''. The final definitional section explores concepts of path sensitivity and how they can be revealed in Markov models.  The definitions provided are presented using probability theory; in addition, each measure has an associated algorithm using matrix operations (excluded from current draft).  Finally an extensive future work section indicates many directions this research can branch into and which methodological, conceptual, and practical benefits can be realized through this suite of techniques.
\end{abstract}

\vspace{.5in}
\newpage

\tableofcontents

\newpage

\section{Introduction}

The idea of tipping points has captured the public's attention from topics as diverse as segregation, marketing, rioting, and global warming. Robustness considerations have extended beyond engineering and ecology to political regimes, computer algorithms, and decision procedures.  Analyzing and exploiting path dependence plays a significant role in technology spread, institutional design, legal theory and the evolution of culture.  However these concepts have not been generally and formally defined and, as a result, the terms' uses across these various applications are hardly consistent.  At times a tipping point refers to a threshold beyond which the system's outcome is known.  Other times `tipping point' is used to describe an event that suffices to achieve a particular outcome, or an aspect of such an event, or the time of such an event.  Another use of tipping points is to label the conditions to which the system is most sensitive.  The idea is frequently tied up with processes such as positive feedback, externalities, sustainable operation, perturbation, etc. Robustness and path dependence share also in this preponderance of senses and this paper aims to elucidate the distinctions among these and other uses of the terms. 

To accomplish this conceptual analysis this paper puts forth formal definitions for each concept (with an implied algorithm) to measure properties of system dynamics.  The analysis utilizes Markov model representations of systems and so definitions of the foundational concepts of system dynamics (equilibrium, basin of attraction, support, etc.) are first provided.  Then various tipping point-related concepts are described, defined, and illustrated with a simplified graphical example.  This treatment is then repeated for robustness-related concepts and then again for a variety of path sensitivities in system dynamics.  

An additional section identifies projects for future work in considerable detail. Through presentations and conversations the techniques presented here have garnered considerable interested from within academics and from government and industry.  The next step is clearly to apply these measures to existing data and models to refine the measures and contribute to science.  A planned extension provides a methodology (and software) to (automatically) generate the state transition representation from observational and model-generated data.  This software tool is necessary for many of the intended and most useful applications of these measures.   There is an additional and unexpected potential conceptual benefit to the robustness formalization presented here for the philosophical study of dispositional properties.  Dispositions such as fragile, soluble, and malleable have long resisted necessary and sufficient conditions to distinguish them from categorical properties (like red, liquid, and square).  The formal definition of one dispositional property (robust) may shed some conceptual light on how to proceed. Other research thrusts that extend the conceptual and methodological benefits of what follows are also presented.

\section{Background}

This paper applies methods from ergodic theory, network theory and graph theory to whole systems encoded as Markov models to find and measure tipping points, robustness, and path sensitivity in system dynamics.  The fusion of these techniques to this purpose is novel, but certainly there is nothing unprecedented about analyzing systems to find these properties or modeling dynamics with Markov models.  But since this paper represents the first marriage of these two realms, the story thus far must be told as two separate threads.

\subsection{Tipping Points, Robustness, and Path Dependence}

Much previous work in finding and measuring properties of system dynamics has focused on explanation - the answering of the `why' question.  Not surprisingly since these papers, books, and discussions were couched in scientific contexts where a particular phenomenon (or class of phenomena) required explanation.  Each such jaunt into explaining tipping points, robustness, or path dependence was accompanied by a custom-suited methodology capable of generating and detecting that property in the model provided (to answer the `how' question).  These models achieved varying levels of generality, but each was limited by the desire to explain the property in a particular model or context.  This is limiting because in order to explain how a process generates (say) path dependent behavior one has to model that process explicitly.  

The current work is one of pure methodology rather than a purported model of any particular system or causal apparatus.  It is meant to be completely abstract and general and therefore capable of measuring these system properties in any system.  Because it does not model any generating process it cannot address the `why' or `how' questions.  It is not meant to.  This paper answers the `whether' and `how much' questions.  These questions are also asked in previous work, but results could not be compared between models because the methodology was model-specific.  A general methodology provides a framework through which all modelers (and some data analysts) can determine whether and how much of each of these properties of system dynamics obtains $\ldots$ and compare results across models regardless of the generating mechanisms.  The ability to compare measures across systems is achieved through a focus on scale-free measures - measures that do not depend on the size of the system being analyzed.  This framework allows scientists to focus on making appropriate models of their subject-matter by eliminating the burden of figuring out how to measure these properties for their model.  Here I review some previous work that includes the similar measured properties to highlight where this methodology might prove useful.

\subsubsection{Tipping Points}

The term `tipping point' was first coined by Morton Grodzins in 1957 \cite{grodzins57} to describe the threshold level of non-white occupants that a white neighborhood could have before ``white flight'' occurred.  The term continued to be used in this context through the work of Eleanor Wolf \cite{wolf63} and Thomas Schelling \cite{schelling71} who also extended the concept to other similar social phenomena.  Though these researchers had a specific usage with narrow focus, the idea of a critical parameter value past which aggregate behavior is recognizably different spread across disciplines where its meaning and application varied considerably.

Malcolm Gladwell's pop sociology book \textit{The Tipping Point} \cite{gladwell00} has played a significant part in bringing the term to the public's awareness.  The notion of tipping point most frequently used by Gladwell is an event that makes something unusual (such as Hush Puppy shoes) become popular.  More precisely this is a critical value for producing a phase transition for percolation in certain heterogeneous social network structures.  This form of tipping point behavior also appears in the work of Mark Granovetter \cite{granovetter78} and Peyton Young \cite{young03} for the propagation of rioting behavior and technology respectively.  This version of tipping will play only a minor role in what follows, however the fact that the expression has made it into the everyman's conceptual vocabulary boosts the importance of establishing rigorous scientific definitions to disambiguate loose usage.

A recent trend in reports of climate change is to refer to a hypothesized tipping point in global warming and ice cap and glacial melting.  James E. Hansen has claimed that ``Earth is approaching a tipping point that can be tilted, and only slightly at best, in its favor if global warming can be limited to less than one degree Celsius.'' \cite{farrell06}  This usage reflects Hansen's belief that ``Humans now control the global climate, for better or worse.''  Gabrielle Walker states, ``A tipping point usually means the moment at which internal dynamics start to propel a change previously driven by external forces. \cite{walker06}''  It is unclear whether Walker's and Hansen's comments are compatible; the conceptual ambiguity of the terms may be making them talk past each other.  But even if their usage is meaningful within their fields, they fail as general characteristics.  Identifying tipping points (as a property of system dynamics) should not depend on whether humans are in control of system behavior or what is driving the dynamics (even if explaining why those are the dynamics does).  

But not all heretofore definitions of the term `tipping point' have been loose or subject-matter specific.  It is often deployed as a semi-technical term in equation-based models of various sorts.  For example, it can refer to an unstable manifold in a differential equation model, the set of boundary parameters for comparative statistics \cite{rubineau07}, or inflection points in the behavior of functional models.  Each of these uses of `tipping points' conforms to our intuitive sense of the term's meaning and at some slightly higher level of abstraction these tipping behaviors are the same - and match the definitions provided in this paper.  But not all models can be faithfully represented as systems of equations and this limits the usefulness of equation-dependent definitions.  One set of tests we can perform on the compatibility of current analysis is to generate Markov models based on the existing differential equation and comparative static models and then determine whether the definitions provided here identify the same states as tipping points.  Such a project is left for future work.

\subsubsection{Robustness}

Robustness considerations are already a common analysis path for researchers in many fields: ecology, engineering, evolutionary biology, logistics, computer science, decision theory, and even statistics.  Models in these fields are often developed specifically to enhance system robustness, avoid system failures, mitigate vulnerabilities, and otherwise cope with variations in an unpredictable environment.  These previous analyses provide some understanding of what features make certain systems persist and others fail, but there is little in the area of general theory.  One hope is that the present construction of general measures of robustness-related concepts will inform and facilitate the construction of general theories of what systemic features produce these properties of system dynamics.  If it can reveal that robust configurations and dynamics in these varied fields can be captured by a single measure, then we will have taken the first step towards a unified theory of robustness.  

Understanding how social systems can be both simultaneously flexible and strong has garnered increasing interest recently.  In an upcoming book Jenna Bednar investigates how institutional design can affect the robustness of a federalist governing body \cite{bednar08}.  ``By explicitly acknowledging the context dependence of institutional performance, we can understand how safeguards intersect for a robust system: strong, flexible, and able to recover from internal errors.''  Bednar has identified the properties that make institutional system robust (compliance, resilience, and adaptation) in a way that is somewhat specific to the subject matter.  That is beneficial and to be expected for explaining and improving the robustness of political institutions. Such an analysis stands to gain from the conceptual refinements derived from formal measures of multiple robustness-related features of system dynamics.  Bednar's work especially underlies the thought that understanding many systems of interest requires more than traditional equilibria analysis; the dynamic nature of dissipative structures (see example \ref{DissipativeStructure}) requires new notions of stability, resilience, and robustness that I hope to help inform through the provided measures.

Thomas Sargent has made extensive use of principles from robust control theory in his analysis of monetary policy and pricing (and other topics).  The sense of `robustness' used in robust control theory is a gap between modeled levels and actual levels of parameters.  It is used to formalize misinformation, uncertainty, and lack of confidence in agents' knowledge and, more generally, to facilitate high levels of performance despite errors and in known less-than-ideal conditions.  This sense of robustness applies across the decision theoretic sciences and planning literature (e.g. the work of Rob Lempert and company \cite{lempert02}).  However, not all robustness analyses are to cope with uncertainty.

In genetics the term robustness refers to a species' consistency of phenotype through changes in the genotype.  Robustness can be considered at two levels: 1) through how much mutation is a member of a species viable and 2) how much genetic variation is required to transform a species' physical characteristics.  The first level takes genetic profiles of organisms and determines which can survive to reproductive age and which cannot (or are sterile).  The number of genotypic variations that remain viable is a measure of the species' robustness $\ldots$ according to that usage.  On the evolutionary time scale we wish to understand how incremental genetic drift is responsible for large phenotypic variations over time.  Walter Fontana has demonstrated that a network of neutral mutations (ones that do not affect fitness) can sufficiently explain the observed punctuated equilibria (see example \ref{PunctuatedEquilibria}) in species evolution \cite{fontana03}.  Though fitness may remain neutral through some genetic variation, the connection between fitness change and phenotype change is strong.  A model that tracks fitness through genetic variations could then approximately measure how robust each stage in the evolutionary progression is.  It is clear that these two concepts of robustness are distinct; and they are both distinct from control and decision theories' usage as well.  

We can add robustness measures from statistics and computer science to the variety of senses that `robustness' can take.  In computer science an algorithm, procedure, measure, or process is robust if small changes (errors, abnormalities, variations, or adjustments) have a proportionally small affect on the algorithm, procedure, measure, or process.  The time complexity of two algorithms may change in different ways.  Algorithm A may require one step per input ($O(n)$) and Algorithm B may require one step per two to the power of the input size ($O(2^n)$); in this case algorithm A is more robust to changes in input size.  Statistical robustness is either when an estimator performs sufficiently well despite the assumptions required by that estimator being violated or when (like in computer science) a measure changes little compared to changes in the input.  For example, the median is a more robust measure than the mean because to alter the median a data point has to cross the median point, whereas any input value change will change the mean's value.

And there are more variations in information theory, data security, engineering, law, ecology, and just about every field has their own version of robustness.  They all share certain high-level conceptual commonalities, but differ in their details and criterion for application.   The definitions below produce necessary and sufficient conditions for the application of several robustness-related concepts.  There are many different ways in which systems can cope with variation are each has its own definition.  This level of refinement (combined with a very inclusive Markov modeling technique) may be able to bring discussions of robustness in the different fields to a single table and foster inter-discipline research.

\subsubsection{Path Dependence}
 
The level of interest in explaining path dependent processes has risen in recent years.  This is in part due to an overall appreciation for the importance of such dynamics in complex systems across domains.  This is also partly due to a natural need to explain observed path dependent phenomena such as convention lock in (e.g. QWERTY keyboards), climate change, and political instability.  And another part is due to an increased prevalence in models wherein path dependence could potentially be formally measured.  Each technique to measure path dependence requires a definition to characterize path dependence in a manner measurable by the formal machinery presented.  Because previous work focused on explaining path dependence through demonstrating sufficient mechanisms to generate it, and this paper's goals are to provide a general system-level definition and way to measure it, the current work will only barely touch on previous research.  Yet insofar as the definitions should be compatible it is worth taking a look at previous, recent formal definitions of path dependence in the literature. 

According to James Mahoney \cite{mahoney00} there are three basic characteristics of path dependence in the social science literature. The first type of path dependence appreciates sensitivity to events that take place in the early stages of a sequence of events. Secondly, there are some historical events early in the sequence that are not explained by prior events. Finally, the sequence of events exhibit some kind of ``inertia'' culminating in an equilibrium-type outcome.  Though later analyses (including this one) deny these characteristics in favor of other ones, this work does demonstrate some popular thoughts in the formal literature on path dependent processes.  Paul Pierson presents one definition that contrasts with Mahoney's.  In Pierson's work path dependence describes how early historical events act to select among multiple possible equilibria.  This equilibria selection process, however, is due to exogenous shocks to the system, a characterization which does not seem necessary for a definition of path dependence.  It does ring true that if only one equilibria-like outcome is possible then system ought to be characterized as \textbf{path independent} (or at least that the outcome be characterized as path-independent).  

Kollman \& Jackson have a variant of Page's definition of path dependence (see below) - one that applies only to specifically parameterized dynamical system models and requires ``very specific and stringent conditions in order for there to be path dependence \cite{kollman08}.''  Path dependence is revealed when a certain time-varying autoregressive parameter is (or converges to) one as the system's dynamics progress through specified shocks.  One of their results is that ``Several steps ought to be taken prior to proposing that some process is path dependent. Such a proposal should be based on rigorous analysis of the process at issue. We cannot offer a set of computational tools that can be used `off the shelf.'  There is no substitute for theoretical modeling appropriate to the system under study as a first step.\cite{kollman08}''  Indeed, building a data-generating model is a necessary first-step for measuring path dependence in most cases.  However, if data satisfies certain properties (described in the next section) then the methodology presented in this paper can be used ``off the shelf'' to measure the degree of many forms of path dependence (see \textit{Automatically Generating the Markov Model} in the future work section for more details).\footnote{Data from the social sciences, however, is least likely to fulfill these necessary properties (for reasons outlined in the opening of the Kollman \& Jackson paper) and will therefore usually require the building of a model to generate the Markov model upon which this analysis can be run.}  If one is primarily interested in finding out whether and how much path dependence a process produces, however, the model need not be specified with the details required by their methodology.  Such model-level tinkering is only necessary to provide an explanation of whence the path dependence.  But if we accept that there are very many mechanisms that can generate path-dependence and that this feature of the mechanisms will be revealed through the data they produce, then we can ignore the mechanisms' specifics when measuring their degrees of path dependence.  This is the approach taken here.

The work of Scott Page in identifying the types and causes of path dependence \cite{page06} is closest in flavor to the work presented here.  Though again an example of mechanism identification, the motivation identified in his introduction applies equally to this paper and is worth including here. 
\begin{quote}
Attempts to extend what is meant by path dependence reflect a need for a finer unpacking of historical causality. We need to differentiate between types of path dependence. The way to do that is with a formal framework. An obvious advantage of having such a framework is that we can conduct empirical analyses and discern whether the evidence supported or refuted a claim of the extent and scope of the sway of the past. That said, empirical testing of a framework of causality is far from the only reason for constructing a framework for modeling historical forces. Formal models discipline thicker, descriptive accounts \cite{gaddis02}. By boiling down causes and effects to their spare fundamentals, they enable us to understand the hows and whys; they tell us where to look and where not to look for evidence. They also help us to identify conditions that are necessary and/or sufficient for past choices and outcomes to influence the present. \cite{page06}
\end{quote}
The differences that Page recognizes in the underlying forces driving law making, pest control, or technology choice are points well taken.  Insofar as the different forces generate differentiable system dynamics the techniques of this paper will be able to identify and precisely measure how much and in what ways past states influence the future.  The properties defined through the formal framework presented below is intended to guide scientists towards characteristics of the original system that merit closer examination.

\subsection{Markov Modeling, Network Theory, Graph Theory, etc.}

Markov modeling has a long history in mathematics, engineering, and in applications to fields as diverse as condensed matter physics, genomics, sociology, and marketing.  Techniques to use Markovian processes to uncover information about system dynamics fall within the field of ergodic theory.   Several properties and techniques from standard Markov modeling are employed below.  In their abstract form one can compute features such as the equilibrium distribution, expected number of steps between two states (with standard deviation), reversibility, and periodicity.  These features gain added meaning when interpreted for the system being modeled, but this paper utilizes them as part of defining (and creating algorithms to uncover) interesting system dynamic properties.

Computer scientists have long been analyzing networks in the form of actual communication networks as well as various abstractions from these problems.  They have invented several useful measures and exceptionally well-crafted algorithms to calculate connectedness, load-bearing properties, path switching, transmission speed, and packet splitting and fusion to name a few.  My analysis borrows heavily from this work in terms of algorithms, though each has been repurposed to the abstract Markov model system representation.  Computer science is also the home of finite state machines: mathematical objects that share their states \& transitions structure with Markov models (though state machines are frequently not probabilistic).  The states of a finite state machine represent the internal states of some agent and the transitions represent the behavior rules by which agents change their states.  Few of the techniques invented to analyze finite state machines will apply to this research because few are adapted to purely probabilistic transitions.  

Hardware engineers and their physicist partners have worked out several interesting measures for circuit design problems.  Multiple paths, variable resistance, flow injection, capacitance and many other characteristics of electronic circuits have analogs in the Markov models presented below.  Though these are only partially explored at this stage, future work will look deeply at borrowing techniques from circuit research.

And finally graph theory offers a few useful measures for our purposes, and moreover provides a wealth of definitions for graph structure and node relationships.  Structural properties will play a larger role in future work addressing changes of resolution and in establishing equivalence classes of system dynamics.  Also, many of the features that graph theory identifies have been given alternative definitions that underlie the probabilistic nature of the Markov model analysis.  Few explicit references to previous work in these methodological subjects appear below because the methods used generally fall in the category of common knowledge.  When a specific algorithm or specialized technique is used, a reference is provided.

\section{Motivations and Applications} 

The measures defined here are meant to stand on their own as improvements in our conceptual understanding of the included features of system dynamics.  By differentiating and formally defining these properties of processes we gain both a common vocabulary with which to discuss our models and a detailed typography of behavior to include and detect in system models.  Many of the applications I have in mind are to include these measures in constructive models across multiple disciplines where the models are iterated with multiple initial settings and/or have stochastic parameters.  These include game theoretic models, network models, physical models, and the whole gamut of models which may be considered agent-based.

Certain static data sets, the sorts collected by surveys, are also analyzable via this methodology.  The data must satisfy certain criteria to be thus analyzable - basic properties it must have before even considering the statistical issues involved with particular applications.  There must be

\begin{enumerate}
 \item Data across time (because we're measuring properties of dynamics)
 \item Repeated system states (so the Markov model isn't deterministic)
 \item Known (or known to be fixed) time between observations
\end{enumerate}

To have repeated system states from collected data we typically will need data from multiple independent trials.  In some cases independence will not be true but will be a useful/necessary approximation.  For example, voter data from each state or county are not truly independent, but each state or county could constitute a separate trial.  In many ways these considerations parallel issues already present in statistical analysis, though the motivation for these requirements are quite different.  Other features of typical statistical analysis will also show up in this technique's application (e.g. correlation, covariance, fixed effects, kurtosis), but they will not be highlighted except where doing so enhances the discussion.

In some cases the data may be recoded or otherwise translated into system states in such a way to as to highlight those features of the system we wish to track and uncover the tippiness, robustness and path sensitivities of.  In some cases this will be a resolution choice, in others this will be a shift to recording Markov states as rate changes, and in others it may be converting fixed-time dynamics into event-drive dynamics.  It will take expertise to determine which, if any, conversion is necessary and a great deal of trial and error to develop that expertise.  Eventually standards will be uncovered as people build proficiency in this methodology.  

A final consideration has nothing to do with the structure and format of the data.  Even if the data is amenable to the analyses presented below, the data may not be observations from a system for which we think robustness or path sensitivity apply.  For example, even if we could run the robustness analysis on voter data, it is not clear what the result would be telling us.  Similarly, given any set of real numbers we can calculate the mean value, but there is no useful interpretation of that mean value for some sets of real numbers (e.g. telephone numbers).  Path sensitivity, however, is something we expect to uncover in poll data and so the amenability of the voter data to the Markov modeling enables this analysis.  The modeling and analysis techniques need to be both possible \textit{and} appropriate, but only a human can determine appropriateness.  

\subsubsection{Autonomic Resource Management}
One promising application of the measures and methods defined in this paper is in developing self-managing large-scale systems - so-called \textit{autonomic systems}.  Autonomic systems apply top-down measures of their internal mechanisms to adapt to changes in resource needs and availability.  Some computer systems, such as internet routers, satellites, autopilots, and explorer robots, use autonomic management programs.  Another major application could be logistics management: the routing of parcels, fuel, food, luggage, etc. to minimize service breakage and costs.  With formal measures of robustness in hand (especially in combination with understanding of path sensitive trajectories) these systems could guide themselves to maintain their function and cope with environmental perturbations using levers uncovered through the analyses presented below.  Because my research extends into each of these research arms the potential gain for having these measures constitutes a significant personal motivation to develop them.

\section{Defining a Markov Model}

At its most basic a Markov model is a collection of states and a set of transition probabilities between pairs of states.  There are several ways to represent a Markov model, but the standard techniques are to 1) model the states as vertices (nodes) and the transition probabilities as weighted edges of a graph or 2) use such a graph's corresponding adjacency matrix or edge list.  Different applications of Markov modeling take different system features as the states, but the nodes in the Markov models used here represent a complete description of a state of the system (see below).  The transition probabilities represent either observed system dynamics or theoretically posited state changes.  Given that states and transitions are defined this way it is clear that the set of states and their transition probabilities are constant for the Markov models utilized in this paper.  

Before beginning the breakdown of the aforementioned phenomena into their various categories, a general typology of state spaces will be helpful.  A system state is a complete set of instantiations of the aspects of the system (values for variables, existence for agents, etc.).  Throughout we will analyze systems with a finite (but possibly arbitrarily large) number of states each with a finite number of aspects. Insofar as some parameters may take on unbounded values (e.g. a continuum of real values) bounding the number of the parameters (i.e. dimensionality of the parameter space) does not ensure a bounded or discrete state space.  The analysis that follows is limited to a finite, discrete state space achieved by binning continuous parameters.    
\begin{define}\label{aspects} A state in the Markov model is a complete specification of the $Q$ aspects of one configuration of the system.
\[ S_i = \{ X_{1(i)}, X_{2(i)}, \ldots X_{Q(i)} \}  \] where $X_{h(i)}$ is the value of aspect $h$ in state $i$ \end{define}
\begin{define}\label{SameState} Two states are represented as one state of the Markov model if all the aspects of the two states are identically valued.
\[ S_i = S_j \Leftrightarrow \forall h \; X_{h(i)} = X_{h(j)} \] \end{define}
An obvious (but still useful) corollary results directly from the truth values of biconditional equivalence.
\begin{cor}\label{Difference} A difference in any aspect marks a different state of the system.  \[\exists h \: X_{h(i)} \neq X_{h(j)} \Leftrightarrow S_i \neq S_j \] \end{cor}

\begin{exa} If our system is an iterated strategic-form game played by six players $$P_i \in \{P1, P2, \ldots, P6\}$$ each with four possible actions $$a(P_i) \in \{a1, a2, a3, a4\}$$ then each state of the system has six aspects and each aspect takes on one of four values. That is $$S_i = \{ a(P1_{(i)}), a(P2_{(i)}), \ldots, a(P6_{(i)})\}$$ and a particular state $S_3$ might be $\{a3, a2, a3, a1, a4, a3\}$. There are $6^4 = 1296$ combinations of four actions for six players, but the Markov model may not include all of them.  Recall that the model is expected to be built from either collected data or a theoretical model so some combinations of aspect values may be unobserved or theoretically impossible or irrelevant. 
\end{exa}

\begin{exa} To analyze Schelling's segregation model \cite{schelling78} we have several options for how we capture the states of the system.  Consider an $8 \times 8$ grid with 50 agents of two types.  We could choose to track the $x$ and $y$ coordinates of each agent in the model which would generate states with 100 aspects (two for each agent).  We could instead track whether each agent is happy with its neighborhood so there would be 50 binary \textit{\{yes, no\}} aspects to each state.  Alternatively the aspects could represent what is in each of the 64 grid spaces with values from \textit{\{empty, agent type1, agent type2\}}.  Note that not all $3^{64}$ combinations are possible because there are fixed numbers of each type of agent and two grid spaces change every transition (agent moving). These and other specifications could be combined with each other and/or measures of the configuration (e.g. how clustered the agents are).  The choice of what to count as the aspects of the states will determine what the measures defined below can reveal. 
\end{exa}

A set of $n$ states is demarked with boldface type: $\mathbf{S} = \{S_1, S_2, \ldots S_n\}$.  The set of all the states in the Markov model is $\mathbf{N}$ which has size $| \mathbf{N}| = N$ -- thus $N$ is also the number of nodes in the graph representation. The state of the system at time $t$ (denoted $s_t$) changes to $s_{t+1}$ in discrete, homogenous time intervals.  State transitions are probabilistic and specified by the system's transition diagram or matrix (see figure \ref{MarkovExample}).  We write the probability of transitioning from state $S_i$ to state $S_j$ as $P_{ij} := P(s_{t+1} = S_j | s_t = S_i)$.  It will later be useful to denote the set of transitions $\mathbf{E}$ and the size of this set as $| \mathbf{E}|$.

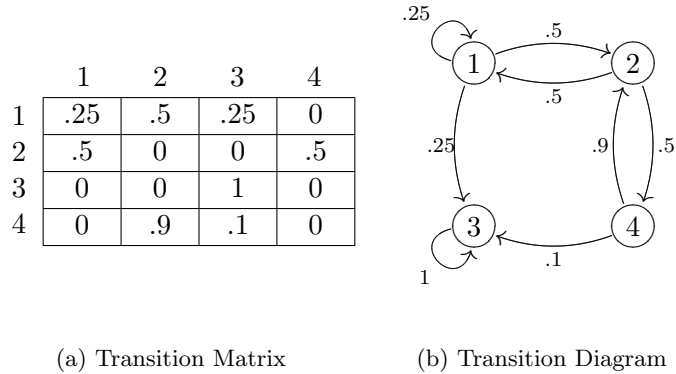
\begin{figure}[!ht]
\centering \caption{Example Equivalent Markov Matrix and Diagram}
\label{MarkovExample}
\begin{center}
\begin{tabular}{M N}

\begin{tabular}{p{2mm} c }
 & \begin{tabular}{P P P P} 1  &  2  &  3  &  4  \\ \end{tabular} \\
 \begin{tabular}{r}1\\2\\3\\4\\ \end{tabular} & \begin{tabular}{|P|P|P|P|}
  \hline
  .25 & .5 & .25 & 0 \\
  \hline
  .5 & 0 & 0 & .5 \\
  \hline
  0 & 0 & 1 & 0 \\
  \hline
  0 & .9 & .1 & 0 \\
  \hline
\end{tabular}\\
\end{tabular}
&
\begin{center}
\UseTips \entrymodifiers={++[o][F]}\xymatrix @=15mm {
1 \ar@(l,u)[]^{.25} \duppertwocell<-3>_{.25}{\omit} \ruppertwocell<3>^{.5}{\omit} & 2 \luppertwocell<3>^{.5}{\omit} \duppertwocell<3>^{.5}{\omit}\\
3 \ar@(l,d)[]_{1} & 4
\luppertwocell<3>^{.1}{\omit}\ulowertwocell<3>^{.9}{\omit} }
\end{center}
\\
{\footnotesize(a) Transition Matrix} & {\footnotesize(b)
Transition Diagram}\\
\end{tabular}
\end{center}
\end{figure}

Following the standard definition from probability theory:
\begin{define}\label{Probability} The sum of a state's exit probabilities must equal one. \[ \forall S \in \mathbf{N} \: \sum_j P(s_{t+1} = S_j | s_t = S_i) = 1 \] \end{define}
The entry in the transition matrix at row $i$ and column $j$ represents $P_{ij}$ and so each row must sum to 1.  
\begin{thm} The probability of a state change equals the probability that each of the aspects of the state changes.  \[ P(s_{t+1} = S_j | s_t = S_i) = P(\forall h \; X_{h, t+1} = X_{h(j)} | X_{h, t} = X_{h(i)}) \]
\end{thm}
This theorem, which relates state transitions back to changes in their constituent aspects, follows directly from definition \ref{SameState} of state sameness and applies to self transitions (i.e. $S_i = S_j$) as well.  It is important to note that the state change probability is not the sum of the aspect change probabilities.  Each single or multiple aspect change produces a distinct, independent state change with its own probability. This property (and others related to aspect changes) will be useful in the discussion of levers below.

\section{Special States and Sets}

To use a Markov diagram to represent system dynamics we will need to define various types of system behaviors in terms of system states, sets of states, and state transitions.  As a preliminary to the common features of system behavior I will present definitions of some structural features that will be utilized.   

\subsection{Paths}

In graph theory a path (of length $\ell$) is typically defined as a set of vertices and edges satisfying the schema $v_0, e_1, v_1, e_2, \ldots, v_{\ell-1}, e_{\ell}, v_{\ell}$ where the edge $e_i$ links the vertex $v_{i-1}$ to $v_{i}$ \cite{lint92}.  Self-transitions, which represent both a lack of change and a change too small to count as a state change, are an important feature of Markov modeling and hence both nodes and edges may be repeated along paths.  So `path' as it is used here is the broader notion sometimes called a `walk' in the graph theory literature.  
\begin{obs} Since there are neither multi-edges\footnote{Multi-edges are multiple distinct edges between two nodes.} nor hyper-edges\footnote{Hyper-edges are edges that connect more than 2 nodes.} in a Markov diagram the set of vertices (or the set of edges) alone is sufficient to uniquely specify a path as long as successively repeated vertices (or edges) are interpreted as self-transitions.  
\end{obs}
A path in a Markov model could be defined as a set satisfying the same schema used in graph theory, but we will use a slightly different definition to make the probabilistic aspects explicit.
\begin{define}\label{Path} A \textit{path} is an ordered collection of states and transitions such that from each state there exists a positive probability to transition to the successor state within the collection.  A path from $S_i$ to $S_j$ denoted $\mathbf{\tilde{S}}(S_i, S_j)$ or $\widetilde{S_i \: S_j}$ is the set of states $\mathbf{S}$ such that 
\begin{itemize}
\item[(i)] $ s_0 = S_i \in \mathbf{S}$
\item[(ii)] $\exists T \; \forall t < T \; P(s_{t+1} \in \mathbf{S} | s_t \in \mathbf{S}) > 0$
\item[(iii)]$ s_{T} = S_j \in \mathbf{S}$
 \end{itemize} \end{define}
This definition establishes necessary and sufficient conditions for $\mathbf{\tilde{S}}$ to be a path, but does not provide a schema for specifying a particular path.  To specify intermediate states (\textit{markers}) for the system to pass through we can write $\widetilde{S_i \: S_j \: S_k}$ to denote a path from $S_i$ to $S_k$ that passes through (at least) $S_j$.  Such a path is merely the conjunction of the two subpaths $\widetilde{S_i \: S_j}$ and $\widetilde{S_j \: S_k}$.  Any number of markers can be thus specified.  The order of the states specified must be satisfied by the path taken, but it does preclude other states from being visited between the marked states.  To specify a long sequence of path markers this paper uses the notation $\mathbf{\tilde{S}}(S_0, S_1, \ldots, S_T)$.  

To completely specify each state along a path we adopt the notation $\overrightarrow{S_i \: \ldots \: S_j}$ for short sequences and $\vec{\mathbf{S}}(S_0, S_1, \ldots, S_T)$ for long ones.  
\begin{thm}\label{ExactPath}An exact path ($\vec{\mathbf{S}}$) satisfies definition \ref{Path} of a path above. \end{thm}
\begin{proof}$\vec{\mathbf{S}}(S_0, S_1, \ldots, S_T) \equiv \bigcup_{t=0}^{T-1} (\overrightarrow{S_t \: S_{t+1}}) \equiv \bigcup_{t=0}^{T-1} (\widetilde{S_t \: S_{t+1}})$ where for each subpath the $T$ in item (ii) of definition \ref{Path} equals 1.
\end{proof}
Though the term `transition' appeared above in the general description of Markov models, it was not defined precisely.
\begin{define}\label{Transition} The \textit{transition} from $S_i$ to $S_j$ is $\overrightarrow{S_i \: S_j}$.  
\end{define}
\begin{define}\label{length}The \textit{length} of a path is the number of transitions taken between the first and last states. \[ \ell( \vec{\mathbf{S}}(S_i, \ldots, S_j) ) := \sum_{t=0}^{T-1} | \{ \overrightarrow{s_t \: s_{t+1}} \}| \] \end{define}
This (possibly overly complicated) formal definition of length simply uses features of the definition of path above, but it is equivalent to the number of edges traversed along the path.
\begin{thm}A path built from a set of states is at least as long as the number of states in the set.
\[ \ell( \vec{\mathbf{S}} ) \geq |\vec{\mathbf{S}}| \]
\end{thm}
This theorem follows from the fact that states may be revisited along the path.  The definition of length suffices for completely specified paths, but the length of merely marked paths can take a range of values depending on the exact sequence of nodes visited.

\begin{define}\label{Cycle} A \textit{cycle} is a path that starts and ends with the same state. \[ \widetilde{S_i \: S_i} \]  \end{define}
\begin{obs} A cycle of length one is a self-transition.  
\end{obs} This observation follows directly from the definition of a transition and theorem \ref{ExactPath}.
\begin{define}\label{ElementaryPath} An \textit{elementary path} is a path that visits each state within the path exactly once \end{define}  
From this definition it is clear that an elementary path is a path with no cycles - including no self-transitions.  A restriction to elementary paths is particularly helpful for ascertaining certain features (e.g. path existence) because simple algorithms exist for them and because of the following property.
\begin{obs} The length of an elementary path equals the number of states in the path. \[ \ell( \vec{\mathbf{S}}_{elem} ) = |\vec{\mathbf{S}}_{elem}| \] \end{obs} Though this fact obviously follows from the definition of path with the exclusion of cycles, providing a formal definition is trivial and obvious and will be omitted.
\begin{figure}[!ht]
\centering \caption{Example Paths and Cycle}
\label{PathCycleFigure}
\begin{center}
\includegraphics{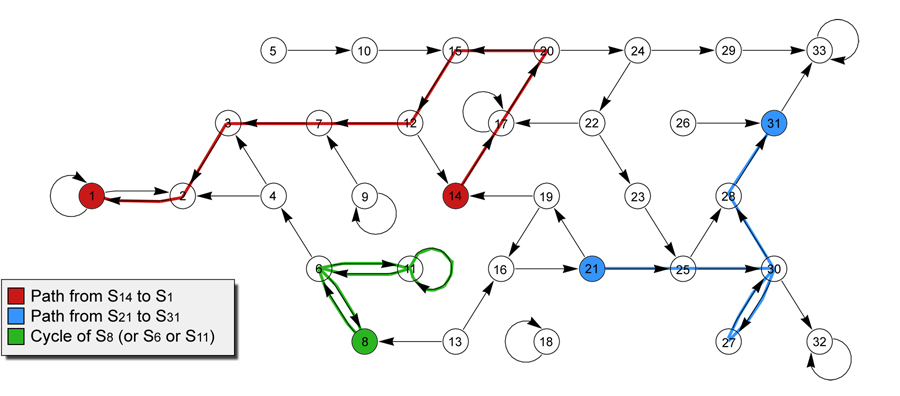}
\end{center}
\end{figure}

Graph and network theorists have developed a great many algorithms for finding paths, calculated their lengths, and measuring properties germane to their application in those fields.  Some of those will come up later in measuring properties of system dynamics, but the definitions and theorems presented above will suffice to move forward in examining our Markov models.  

\subsection{Landmarks in System Dynamics}

This subsection provides definitions and sketches of some algorithms for common structural properties of Markov models used to represent system dynamics.  Many of these features have existing definitions in terms of matrix operations or limiting distributions, but this paper will present alternative definitions in many cases.  The definitions will make focused use of the finitude of the Markov models, the resolution of the states, and the granularity of the probability measurements.  My motivation for the alternative definitions is to facilitate clear intuitions about a system's processes and how to measure them precisely.

\begin{define}\label{Equilibrium} A system state that always transitions to itself is called an \textit{equilibrium} or \textit{stable state}.\footnote{As will be defined formally in definition \ref{StateStability}, stability refers to a tendency to self-transition.  Hence an equilibrium is equivalent to a \textbf{fully} stable state.} An equilibrium $e_i$ is a state $S_i$ such that \[ P(s_{t+1} = S_i | s_t = S_i) = 1. \] \end{define} 
In some cases a set of states plays a role similar to that of an equilibrium.  \begin{define}\label{Orbit} An \textit{orbit} is a set of states such that if the system enters that set it will always revisit every member of the set and the system can never leave that set.  $\mathbf{S}$ is an orbit if 
 \[ \forall i \: \nexists h \: P(s_{t+h} = S_i \in \mathbf{S} | s_t \in \mathbf{S}) = 0 \] and \[ \forall i \: \forall h>0 \: P(s_{t+h} = S_i \not \in \mathbf{S} | s_t \in \mathbf{S}) = 0. \] \end{define}  
\begin{define}\label{Oscillator} An \textit{oscillator} is an orbit that is also a cycle. \end{define}
\begin{thm}Equilibria cannot be proper subsets of an orbit\end{thm}
\begin{proof} Assume $e_i \in \mathbf{S}_{orbit}$.  From definition \ref{Orbit} $\nexists h$ such that $P(s_{t+h} = S_i \in \mathbf{S}_{orbit}| s_t  \in \mathbf{S}) = 0$.  Since $e_i$ is in $\mathbf{S}_{orbit}$ definition \ref{Orbit} implies $\exists h P(s_{t+h} = e_i | s_t \in \mathbf{S}) > 0 $.  Definition \ref{Equilibrium} implies that $\forall h \: P(s_{t+h} = e_i | s_t = e_i) = 1$ which further implies that $\forall h \: P(s_{t+h} \neq e_i | s_t = e_i) = 0$.  For all $S_j \neq e_i$, $\forall h \: P(s_{t+h} = S_j | s_t = e_i) = 0$.  Therefore by contradiction with definition \ref{Orbit} $e_i \not \in \mathbf{S}_{orbit}$. \end{proof}
\begin{obs} Given the definitions of equilibrium and orbit above it is clear that an equilibrium is an orbit of just one state. \end{obs} 
\begin{define}\label{Attractor} An \textit{Attractor} (denoted $A_i$) is either an equilibrium state or an orbit of the system.\footnote{Though different formulations, the above definitions identify the same states as the standard equilibrium distribution definitions from ergodic theory.  The equilibria and orbits are those states with positive probability after a ``long period of time''. Given a Markov model represented by the transition matrix $M$ the attractor states in $\mathbf{A}$ are those where $P(M^T(S_i)) > 0$ as $T \rightarrow \infty$.  Or in a more computable formulation, given a desired degree of significance for the probability measures $\exists T^{*}$ such that $P(M^{T^{*}}(S_i)) > 0$ if and only if $S_i$ is an attractor state. Although it becomes a little more complicated if one wants to isolate the individual equilibria and orbits. Techniques (such as this one) for measuring properties (whether mathematical or computational) will be presented in an appendix in future versions and are occasionally referred to in this text or in footnotes.}
\end{define}
The choice of a resolution determines whether an orbit appears as an equilibrium or \textit{vice versa}.  In cases where attractor avoidance is the aim of the model (see tipping and robustness below) we can collapse orbits into a single attractor state without loss of information.  For this reason I will use `$A_i$' as if it were a single state except in cases where it being an orbit affects the analysis.

In later sections we will encounter the idea of a dissipative structure (see example \ref{DissipativeStructure}) for which equilibria analysis is inappropriate.  It is not the case that these systems fail to have attractors, it's just that the goal of such systems is to remain in continual flux and avoid equilibria and other ``point attractors'' (i.e. attractors that incorporate a small percentage of the total number of states).  Given the definitions established above every Markov model must have some set of states satisfying the conditions for being an attractor.
\begin{thm}\label{MinAttractor} Every system has at least one attractor. \[ \exists A_i \subseteq \bigcup_{i=1}^{N} S_i \]\end{thm} \begin{proof} Assume that there exists a system with no attractors.  If $\mathbf{S}$ is not an orbit then by definition \ref{Orbit} for some $S_i \in \mathbf{S}$ it must be the case that either 1) there is some time in the future after which $S_i$ does not get visited or 2) some state outside $\mathbf{S}$ gets visited.  Let $n$ be $|\mathbf{S}|$.  For any $n$, case (1) implies that any orbit must be smaller than $n$ and case (2) implies that there must be at $n+1$ states in the system. For $n = 1$ it is not possible for the orbit to be smaller so any orbit must contain at least $n+1 = 2$ states.  By induction on $n$ it must be the case that if there are no orbits of size $1 \leq n \leq N$ ($n = N$ is the whole system) then any orbit must be of size $N + 1$ which is impossible.  That contradicts the assumption that there exists a system with no attractors, so every system must have at least one attractor.
\end{proof}

\begin{define}\label{Basin} Those states from which the system will eventually move into a specific attractor are said to be in that attractor's \textit{basin of attraction}.  The basin of $A_i$ or $\mathbf{B}(A_i)$ is a set of states $\mathbf{S}$ such that\footnote{Because the state space is finite and there is a limit to the granularity of the probability measures this definition suffices without needing to take $h \rightarrow \infty$.  See previous footnote for more details.}
\[ \exists h \geq 0  \: P( s_{t+h} = A_i | s_t \in \mathbf{S}) = 1  \] \end{define}
Some systems may spend a great deal of time in a basin of attraction before reaching the attractor located within it thus making system behavior in the basin similar to an orbit itself (also note for the discussion of robust sets below).  In such cases it is sometimes helpful to utilize the following property to describe and make inferences about system behavior.
\begin{obs}Once in a basin of attraction the system can never leave it. \[ \forall h \: P(s_{t+h} \in \mathbf{B} | s_t \in \mathbf{B}) = 1\]  \end{obs}

\begin{define}\label{Support} The \textit{support} of a state (also known as its \textit{in-component}) is the set of states which have a path to it.   The support of $S_i$ or $\mathbb{S}(S_i)$ is the set of states such that
\[ \forall j \: S_j \in \mathbb{S}(S_i) \Rightarrow (\widetilde{S_j \: S_i}). \] \end{define}  
We can expand this definition to the support of a set of states $\mathbf{S}$ as the union of the supports of the members of $\mathbf{S}$.  Some facts relating these features of system dynamics are clear from the above definitions.
\begin{obs} An attractor is a subset of its basin of attraction, and an attractor's basin of attraction is a subset of its support. \[ A_i \subseteq \mathbf{B}(A_i) \subseteq \mathbb{S}(A_i) \]
\end{obs}
\begin{obs} The equilibrium may be the only member of either its basin or support, but if it is the only member of its support then it is disconnected from the rest of the graph (e.g. $S_{18}$ in figure \ref{EquilBasin}).  \end{obs}

\begin{figure}[!ht]
\centering \caption{Attractors, Basins of Attraction, and Support}
\label{EquilBasin}
\begin{center}
\includegraphics{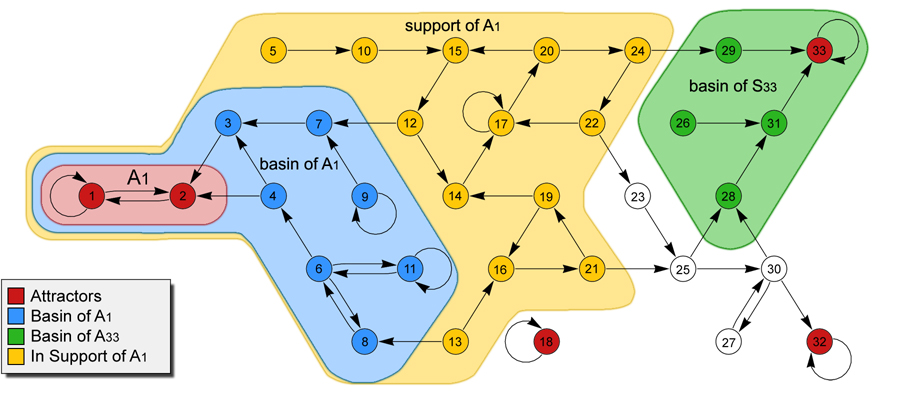}
\end{center}
\end{figure}

The \textit{indeterminate states} of a system, ones that are not members of a basin of attraction,  convey a wealth information about the system's dynamics and its future states. Recall that the whole system may be a single orbit and there may be no indeterminate states.  But if there are multiple attractors, then the indeterminate states are the ones in multiple attractors' supports.  \begin{define}\label{Overlap} The \textit{overlap} of a collection of states (whether attractors or not) is the set of states in all of their in-components (i.e. the intersection of supports). The overlap of $\{S_i, \ldots, S_j\}$, written $\mathbf{\Omega}(S_i, \ldots, S_j)$, is the set of states in
\[ \bigcap_{S_g \in \{ S_i, \ldots, S_j \} } \mathbb{S}(S_g)  \] \end{define}

\begin{figure}[!ht]
\centering \caption{Supports and Their Overlap}
\label{OverlapDiagram}
\begin{center}
\includegraphics{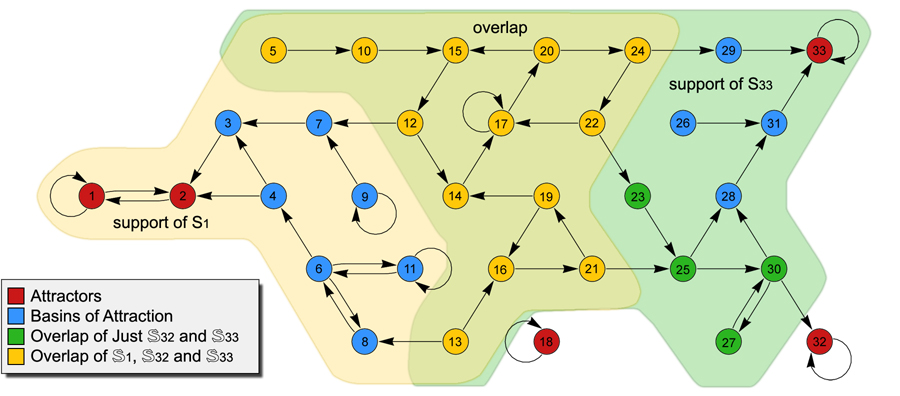}
\end{center}
\end{figure}

Basins of attraction can only contain one equilibrium state or orbit and hence the basins for two different attractors cannot overlap.  Supports of different collections of attractors may or may not have sets of overlapping states.  The overlap states are of interest because these are the states with positive probabilities for ending up in each of the attractors for which the supports are overlapping.
\begin{obs} If no attractors' supports overlap then every attractor's support is just its basin of attraction and the system has a \textit{deterministic outcome}.
\end{obs}
\begin{exa} Models referred to as dynamical systems models typically use sets of differential equations and are typically fully deterministic because the map from $x$ to $f(x)$ is always a function in the strict sense. Because the dynamics are produced from these functions, every initial condition is mapped to a particular equilibrium state.  Certain parameterized systems allow changes in their state maps that can change the number, location, and/or ``strength'' of equilibria.  But for any given value of those parameters, each initial value still has one possible outcome. Many systems can be usefully modeled with such fully deterministic systems, but in this paper we are mostly concerned with the non-deterministic parts of systems because that is where the critical points occur (see next section).
\end{exa}

\begin{define}\label{outdegree} A state's \textit{out-degree} is the number of distinct successor states (states that may be immediate transitioned into).  The out-degree $k_i$ of state $S_i$ equals
\[ | \{S_j: P(\: s_{t+1} = S_j | s_t = S_i) > 0 \} | \] \end{define} 
$S_k$ will be used to denote a neighboring state and $\mathbf{S_k}$ the set of neighboring states.  
\begin{define}\label{indegree} The number of states that can transition into a state is its \textit{in-degree}. \[ | \{S_j: P(\: s_{t+1} = S_i | s_t = S_j) > 0 \} | \] \end{define}
The in-degree measure will be rarely used in what follows (except for algorithms utilizing a reversed Markov model subgraph) and will not need its own symbol.  Because of this asymmetry the term `degree' will refer to a state's out-degree unless otherwise noted.

\begin{define}\label{reach} The \textit{reach} of a state (also called its \textit{out-component}) is the set of states that the system may enter by following some sequence of transitions; i.e. \textit{all possible} future states given an initial state. The reach of $S_i$ or $\mathbf{R}(S_i)$ is the set of $S_j$s such that \[\exists h>0 \: P(s_{t+h} = S_j | s_t = S_i) > 0 \] \end{define} 
\begin{thm}\label{DecreasingReach} Every successor state's reach is less than or equal to the intial state's reach.  \[ \forall i, j \: \overrightarrow{S_i \: S_j} \Rightarrow \| \mathbf{R}(S_i) \| \geq \| \mathbf{R}(S_j) \| \] \end{thm}

\begin{figure}[!ht]
\centering \caption{A State with a Large Reach}
\label{ReachBig}
\begin{center}
\includegraphics{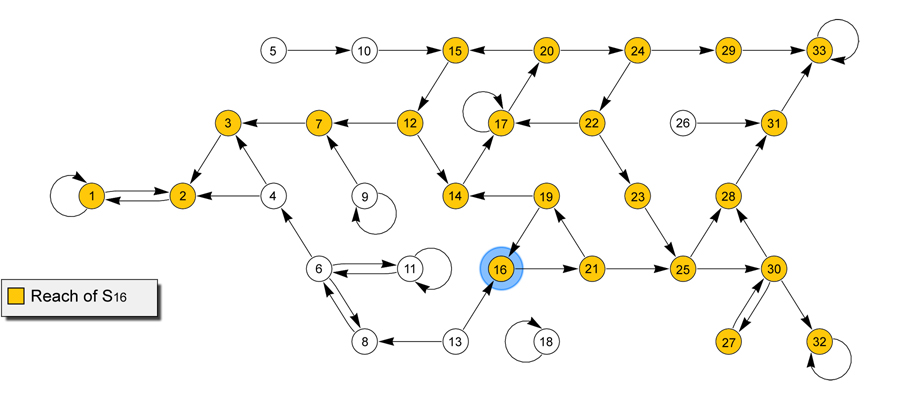}
\end{center}
\end{figure}

\begin{figure}[!ht]
\centering \caption{Two States with Small Reach}
\label{ReachSmall}
\begin{center}
\includegraphics{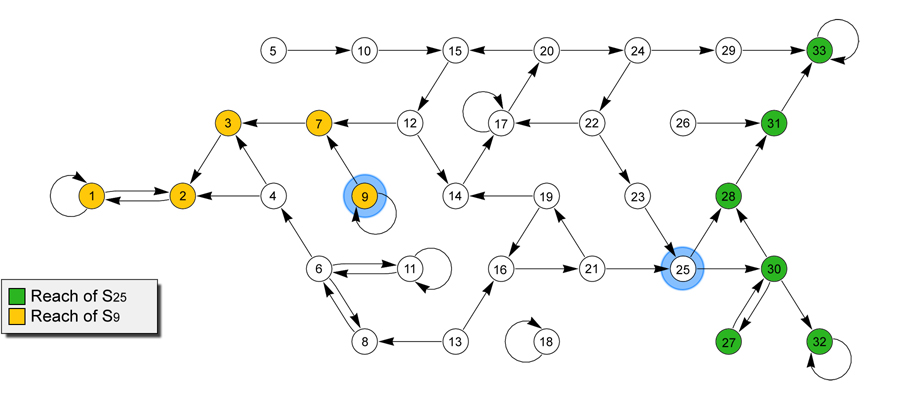}
\end{center}
\end{figure}

\begin{thm}\label{reachpath} $S_j$ is in the reach of $S_i$ if and only if there exists at least one path from $S_i$ to $S_j$. \[ S_j \in \mathbf{R}(S_i) \Leftrightarrow  \exists \widetilde{S_i S_j} \] \end{thm}
This theorem, which can be used as an alternate definition of reach, follows from the definitions of path and reach. 
\begin{cor}An obvious corollary of the definitions of basin and reach is that if the reach of a state $S_i$ includes only one attractor then $S_i$ must be in that attractor's basin of attraction \[\exists ! A_i \in \mathbf{R}(S_i) \Rightarrow S_i \in \mathbf{B}(A_i). \] \end{cor}  
\begin{obs}An attractor state's reach is just the states within the attractor; an equilibrium's reach is itself.\end{obs}

\begin{proof} For all transitions $\overrightarrow{S_i \: S_j}$ it must be the case that $S_j \in \mathbf{R}(S_i)$ from the definition of reach since $h=1$ satisfies $\exists h \: P(s_{t+h} = S_j | s_t = S_i) > 0$.  Furthermore every $S_g$ in the reach of $S_j$ is also in the reach of $S_i$ because for whatever $h$ makes it true that $S_g \in \mathbf{R}(S_j)$, $h+1$ makes it true for $S_i$.  To achieve the inequality it suffices for there to be at least one state in $\mathbf{R}(S_i)$ not in $\mathbf{R}(S_j)$.  It is \textit{prima facie} obvious that it is possible for $\exists S_g$ such that $\overrightarrow{S_i \: S_g}$ and $\nexists \widetilde{S_j \: S_g}$. \end{proof}
Theorem \ref{DecreasingReach} generalizes to all paths (which is just a sequence of transitions) so that reach never increases as the systems transitions along any path.  This property relies on the fact that the transition structure is fixed for the Markov models used in this paper; future work will relax this requirement and the theorem does not necessarily hold for models with changing dynamics.

A \textit{strongly connected component} of a directed network is a set of vertices such that there is a path from every vertex in the set to every vertex in the set (including itself).  We will find the same concept useful, but this paper adopts a different name for it.
\begin{define}\label{core} A \textit{core} of a set is a subset wherein every member of the subset is in the reach of every member of the subset.  The core of some set $\mathbf{S}$ is written $\mathbb{C}_{\mathbf{S}}$ and is a subset satisfying the condition
\[ \bigcap_{S_i \in \mathbf{S}} \bigcap_{S_j \in \mathbf{S}} S_j \in \mathbf{R}(S_i)  \] \end{define}
Some sets will have multiple cores -- the set of $\mathbf{S}$'s cores can be called $\mathbf{S}$'s \textit{mantle}.  
\begin{obs}\label{ReachCycle}Every cycle within $\mathbf{S}$ is (at least part of) a core and every state in a core is in at least one cycle. \end{obs}  
\begin{obs}\label{ReachCore}Every state in a core has the same reach.\end{obs}  

\begin{figure}[!ht]
\centering \caption{The Cores of an Arbitrary Set}
\label{CoreOfS}
\begin{center}
\includegraphics{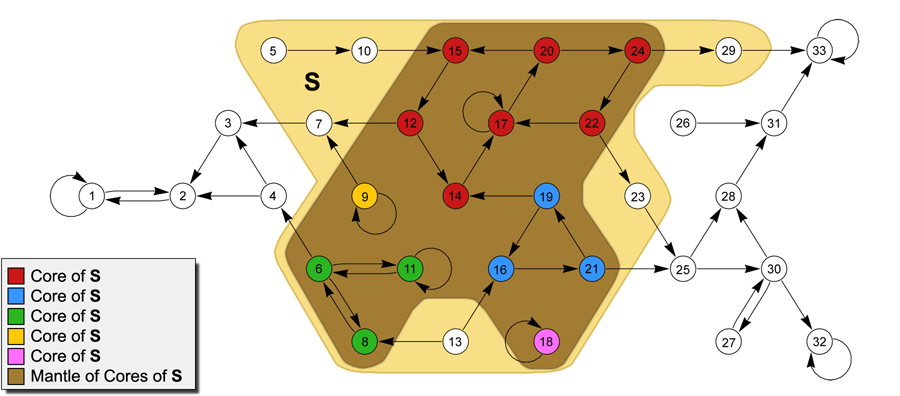}
\end{center}
\end{figure}

The states on the boundary of a set are a useful set of states to identify.  
\begin{define}\label{Perimeter} The \textit{perimeter} of a set, $\mathbb{P}(\mathbf{S})$, is a collection of those states in the set that may transition to states outside the set.  That is, $\mathbf{S}$ such that
\[ P(s_{t+1} \not \in \mathbf{S} | s_t \in \mathbf{S}) > 0 \] \end{define}

\begin{figure}[!ht]
\centering \caption{The Perimeter States of an Arbitrary Set}
\label{PerimeterDiagram}
\begin{center}
\includegraphics{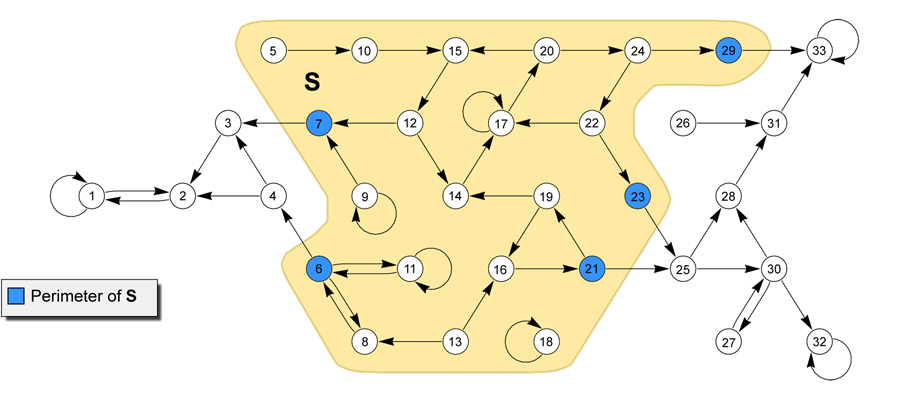}
\end{center}
\end{figure}

In keeping with the core and mantle analogy, the perimeter states of the mantle of $\mathbf{S}$ will be referred to as $\mathbf{S}$'s \textit{crust}.  Perimeter states themselves, without further specification, describe one commonly deployed (though weak) concept of tipping points, although we will see in the next section that specifying different base sets produces different types of tips.

In the next section we apply the above definitions in different combinations and different contexts to identify various system behaviors.  Most extant systems analysis focuses on equilibria, but a lot of interesting behavior happens away from equilibrium. In the indeterminate states of a system we cannot know precisely which states the system will reach or which state it will be in at a given time in the future, we only know probability distributions over the future states.  But by understanding a system's behavior we might know whether some particular change facilitates a specific outcome or path through the dynamic - this may be helpful information.  Discerning these sorts of facts about system dynamics can increases one's information about the system (in both the technical and colloquial senses).  Considerations such as this one are the building blocks of the formal theory of tipping points, robustness, and related phenomena immediately to follow.

\section{Tipping Phenomena and Related Concepts}

Using Markov models and the states and sets defined above as a springboard, this section defines and briefly describes several terms related to the concept (or more to the point, concepts) of tipping points.  Critical phenomena and tipping points of various kinds share the defining feature that (for whatever reason) behavior is different before and after some transition.  Behavior in this analysis is just the properties of systems dynamics.  There are, of course, many ways in which the properties of system dynamics can differ and each way is a different kind of tipping phenomenon.

\subsection{Levers}
Recalling that state changes occur if and only if there is a change in some aspect of the initial state (see corollary \ref{Difference}), our analysis of tipping phenomena starts with state aspects.  \begin{define}\label{Levers} The \textit{levers} of a state are the aspects of a state such that a change in those aspects is sufficient to change the system's state.  The levers of $S_i$, denoted $L(S_i)$, is \[\bigcup_{j=1}^{k} \bigcup_{h=1}^{Q} X_{h(i)} \neq X_{h(j)}.\]\end{define} So given a state in the Markov model, the levers are those aspects of the state that are different in any neighboring state and each such $X_h$ is a distinct lever.  In some cases neighboring states will differ by more than one aspect.  In those cases the respective element of the set $L(S_i)$ will be a list of all the aspects that need to change as one lever.  
\begin{define}\label{LeverPoint} A \textit{lever point} is a transition resulting from a change in a particular aspect (or set of aspects).  An aspect's lever points is the collection of transitions that a change in that aspect (or those aspects) alone generates.  The lever points of $X_h$ is the set of transitions created by \[ \bigcup_{\overrightarrow{S_i \: S_j}}^{\mathbf{E}} X_{h(i)} \neq X_{h(j)} \] \end{define}
Levers and lever points work complementarily: for levers we pick a state and find the aspects that change and for lever points we pick the aspect and find the state changes it produces.  

It is occasionally helpful to refer to the aspect change(s) that generate a specific transition.  
\begin{define}\label{LeverSet} $L(\overrightarrow{S_i \: S_j})$ symbolizes the \textit{lever set} of $\overrightarrow{S_i \: S_j}$: the aspect or aspects that differ between $S_i$ and $S_j$. \end{define} 
In some applications we will be interested in how many aspects change for a transition. \begin{define}\label{Magnitude} The \textit{magnitude} of the lever set of a specific transitions is $|L(\overrightarrow{S_i \: S_j})|$.\end{define}
Though levers as they are defined here do not depend on the ability to control that aspect, the choice of `lever' for this concept is motivated by the realization that in some models control of some aspects is available.  One may be performing a tipping points analysis precisely because one is choosing levers to bring about one state versus another (or agents within the model may be choosing).  

\begin{exa}\label{PolicyChange}Imagine a model wherein each aspect is a variable representing some part of a policy (e.g. amount of money spent on each line item).  Each aspect change has an associated cost (legal, bureaucratic, time, etc.). The modeler may be trying to determine the lowest cost, feasible route from the current policy to some desired policy; or perhaps to determine how far policy can be changed on a specific budget.  The cumulative magnitudes of lever sets along a path may adequately approximate such a cost measure. In general the sum of the magnitudes along a path is a rough measure of how difficult it is for the system to behave that way. Techniques from circuit design applied to the Markov model may be gainfully applied to such models. \end{exa}

In some contexts we may wish to know how much change an aspect is responsible for across the system's dynamics.
\begin{define}\label{Strength}The \textit{strength} of a lever is the sum of the probabilities of all transitions that result from changing that lever. The strength of $X_h$ is equal to \[ \sum_{\overrightarrow{S_i \: S_j}}^{\mathbf{E}}P(\overrightarrow{S_i \: S_j}| X_{h(i)} \neq X_{h(j)}). \]  \end{define}
This measure is not scale-free since the sum depends on the number of transitions in the Markov model, but it is useful for comparing levers within a system.  The strength measure could be used, for example, to determine which aspects to control to maximize (or minimize) one's ability to manipulate the system.  It could also be associated with a cost of letting that aspect vary over time.  We will revisit levers below in other forms as they apply to other measures of system dynamics.

\subsection{Thresholds}

In some cases we are interested not just in which aspects change through a transition but also in the \textbf{values}\footnote{Recall that many of the things that can be included as aspects of states are not numeric parameters and so what counts as a ``value'' for that aspect is meant to be interpreted broadly.} of levers at transitions.  
\begin{define}\label{Threshold} A \textit{threshold} or \textit{threshold point} is a particular value for a lever such that if the value of the aspect crosses the threshold value it generates a transition. So $x$ is a threshold value of $\overrightarrow{S_i \: S_j}$ if \[  X_{h(j)} \neq x \textrm{ and } X_{h(i)} = x. \] \end{define}
This definition can be applied \textit{mutatis mutandis} for a set of values $\{x\}$ for a lever set which can be distinguished by the name \textit{threshold line} when appropriate. 

If there are multiple states with transitions crossing the same threshold value then knowing that information refines our understanding of the lever's role in system dynamics.  Thus determining the threshold value for one transition is merely a means to the end of determining the strength of the levers with that threshold.
\begin{define}\label{ThresholdStrength} The \textit{threshold strength} of $x$ is the strength of the levers for which $x$ is the threshold value: \[ \sum_{\overrightarrow{S_i \: S_j}}^{\mathbf{E}}P(\overrightarrow{S_i \: S_j}| X_{h(j)} \neq x \textrm{ and } X_{h(i)} = x). \] \end{define} If a particular value for a particular aspect plays a large role in system dynamics then crossing that threshold is another oft-used version of ``tipping point'' behavior.

These definitions of threshold and threshold strength only require that the end state's value be different from the start state's value.  In common usage, however, thresholds establish different and separate boundary values for ascending and descending values.   If a threshold only affects system dynamics in one direction then we can determine that from the Markov model using the following definitions.
\begin{define} An \textit{upper bound threshold} of $\overrightarrow{S_i \: S_j}$ is a value $x$ such that \[  X_{h(j)} > x  \textrm{ and } X_{h(i)} = x. \] \end{define}
\begin{define} A \textit{lower bound threshold} of $\overrightarrow{S_i \: S_j}$ is a value $x$ such that \[  X_{h(j)} < x  \textrm{ and } X_{h(i)} = x. \] \end{define}
The threshold strength measure can be adapted to these ascending and descending definitions in the obvious ways.  Sets satisfying these definitions can tell us how frequently crossing that threshold in that direction acts as a lever.  If there are multiple states with transitions crossing the same threshold value then knowing that information refines our understanding of the lever's role in system dynamics. 

\begin{exa}This general definition admits examples from many different kinds of systems and can even apply to parts of systems (such as agents).  In Granovetter's model of riot spreading \cite{granovetter78} we can talk of each agent having its own threshold - the number of rioting agents necessary to make each agent join the riot.  This is just the same threshold definition applied to a lever set where the levers happen to be the same feature of each agent.  In Granovetter's model the threshold value is the same in both directions.\end{exa}

We can also talk of thresholds in the properties of the system dynamics that track how the system transitions through states.  Instead of being a value for an aspect within the model, it would be a value for one of the measures defined in this paper.  
\begin{exa}The following might be the case for some system: once the energy level of the current state drops below three the system is at most three transitions away from being in an attractor.  Though having an energy level is not part of the system that the Markov model represents we can associate this property of system dynamics with each state of the system as if it were one of its possible levers.  Then we can explore the relations of specific values of this property of system dynamics to its other dynamics.  Relating values of properties of system dynamics back to aspects within the model will also provide useful information in many cases.  \end{exa}

\subsection{Critical Behavior}

In some system analyses the property of interest is what is available for the future $\ldots$ in the most general terms.   If one does not know much a system's dynamics then even knowing how many states could potentially be transitioned to provides an informational benefit.  The measures below become increasingly refined and detailed, but we start with some simple measures that may suffice for some applications.  
\begin{define}\label{Stretch} A state's \text{stretch} is the number of states in its reach. So the stretch of $S_i$ equals \[ | \mathbf{R}(S_i) |. \] \end{define}
\begin{define}\label{CriticalBehavior} A system dynamic (i.e. a particular state transition) is considered \textit{critical behavior} if and only if it produces a decrease in stretch; that is, critical behavior is any $\overrightarrow{S_i \: S_j}$ such that \[| \mathbf{R}(S_i) | > | \mathbf{R}(S_j) |. \] \end{define}
In addition to identifying the transitions that limit the system's future states, we can also measure how critical the transition is.  Subtracting the end state's stretch from the start state's stretch provides such a measure, but it is not scale-free\footnote{It is not scale-free because both the range of values and the particular value for this measure depends on the total size of the system.} and so cannot be readily compared across different systems. We can normalize the stretch difference with the size of the system to which it is being applied to produce a percentage measure.
\begin{define}\label{StretchGap} The \textit{stretch-gap} of a transition is the change in the percent of the total number of states that can be reached.  This quantity equals \[ \frac{| \mathbf{R}(S_i) |}{N}-\frac{| \mathbf{R}(S_j) |}{N}. \] \end{define}
Because this measure includes the total number of states in the system it clearly is not scale-free either.  However despite this limitation it does provide information about the system's future and is an intuitive way to compare transitions within the same system - even at different resolutions.  As example \ref{DropExample} below demonstrates the stretch-gap reports how much of the system's state space is cut off by each transition and this information could be used, for example, to manipulate system dynamics to prolong system longevity.  We also have an alternative, fully scale-free, measure of the drop in reach across a transition.
\begin{define}\label{Criticality} A \textbf{transition's} \textit{criticality} is one minus the ratio of the start and end states' stretch. The criticality of $\overrightarrow{S_i \: S_j}$ equals \[1- \frac{ | \mathbf{R}(S_j) | }{ | \mathbf{R}(S_i) | }. \] \end{define}
Recall from theorem \ref{DecreasingReach} that from any initial starting point, as the system transitions through its states the sizes of the states' reach are monotonically decreasing.  As a result of that theorem we have the following corollary regarding the range of values for the ratio of reaches.  
\begin{cor}A transition's criticality will be between zero and one.\end{cor}
\begin{proof} Let $a = | \mathbf{R}(S_j) |$ and $b = | \mathbf{R}(S_i) |$ for transition $\overrightarrow{S_i \: S_j}$.  From theorem \ref{DecreasingReach} $a \leq b$.  $a = b$ produces $\frac{a}{b}=1$ which yields a criticality of zero.  For $a < b$ we can decrease $a$ or increase $b$ to find the other bound, but since $a$ and $b$ are natural numbers increasing $b$ is the better approach. Using a well-known mathematical fact suffices for finding the other bound: $\forall a \: \lim_{b \rightarrow \infty}\frac{a}{b} = 0$. \end{proof}  
\begin{obs}Transitions within a cycle (which includes self-transitions) always have zero criticality and zero stretch-gap.\end{obs}
This observation follows from observations \ref{ReachCycle} and \ref{ReachCore} and the definition of a core.  The concept of criticality agrees with this measure insofar as any transition that has no affect on what states may be visited in the future should not be a critical transition.  

\begin{figure}[!ht]
\centering \caption{Stretch-Gap and Criticality Measures}
\label{CriticalityExample}
\begin{center}
\includegraphics{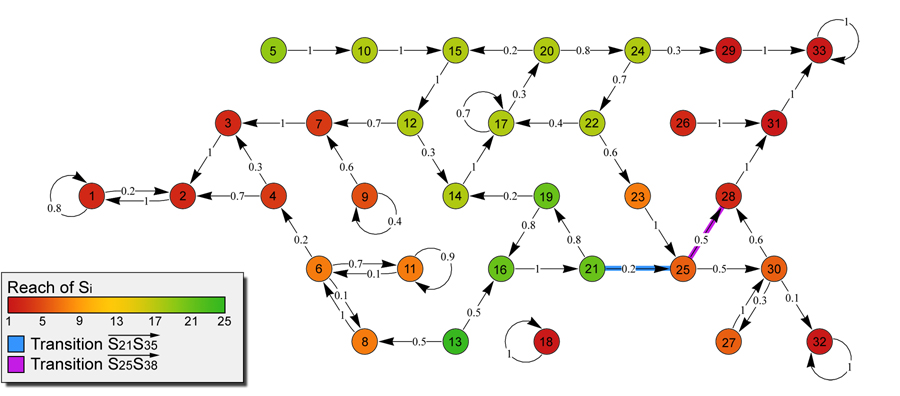}
\end{center}
\end{figure}
\begin{exa} \label{DropExample}
The system represented as figure \ref{CriticalityExample} has thirty-three states in total.  Each one is color-coded by its stretch.  Compare the patterns in color to the attractors, basins, and support in figure \ref{EquilBasin} and the overlap in figure \ref{OverlapDiagram}.  Stretch alone, though a rather simplistic measure, works decently to partition the system dynamics into regions of similar behavior.  Stretch-gap performs well as a discriminator that groups states into these regions in a way similar to spectral analysis for detecting community structure in networks \cite{clauset08}.  To wit, stretch drops between zero and four states within any basin or overlap, but drops four to eighteen states crossing a boundary (in this example).  This is not completely reliable, of course, but for many systems this simple technique may provide all the information required; and it may be the best one can do with available data.

$S_{21}$ has a stretch of 21, $S_{25}$ has a stretch of 6, and $S_{28}$ has a stretch of 2.  There are 33 states in this system so the stretch-gaps of  $\overrightarrow{S_{21} \: S_{25}}$ and $\overrightarrow{S_{25} \: S_{28}}$ are 45.45\% and 12.12\% respectively.  That means that 45.45\% fewer of the system's states can be reached after the $\overrightarrow{S_{25} \: S_{28}}$ transition.  We can also use this to determine the stretch-gap of $\widetilde{S_{21} \: S_{28}}$ as 57.58\% regardless of the particular path taken. Only the start and end states' stretches are necessary to calculate this, but the result is always equal to the sum of the stretch gaps of each transition taken.  

Let's compare these figures to the criticality of the same transitions. The criticality of $\overrightarrow{S_{21} \: S_{25}}$ is 0.714 and the criticality of $\overrightarrow{S_{25} \: S_{28}}$ is 0.667.  That means that the system only has 71.5\% of the possible future states in state $S_{25}$ as it did in $S_{21}$.  A composite measure is also possible for the criticality of $\widetilde{S_{21} \: S_{28}}$. It can be calculated just using the start and end states' stretches, using the standard percentage of a percentage of a percentage $\ldots$ calculation. \footnote{The value equals the iterated sum of the previous transition's criticality and the product of the transition criticality with the previous transition's criticality's complement} So the criticality of $\widetilde{S_{21} \: S_{28}} = 0.714 + 0.667 \cdot (1 - 0.714) = 0.905. $
\end{exa}

These measures above are intended to be just rough measures useful in certain limited contexts and when information about the system is limited. For starters, these measures only consider only the structure of the Markov models, not the probabilities.  Also, they apply to transitions rather than states.  
\begin{define}\label{StateCriticality} The \textit{criticality} of a \textbf{state} is the probabilistically weighted sum of the criticality of all the transitions from that state. So to find the criticality of $S_i$ we calculate \[ \sum_{j=1}^{k} P_{ij} \left( 1 - \frac{ | \mathbf{R}(S_j) | }{ | \mathbf{R}(S_i) | }\right)  \] where by convention that is the sum over $S_i$'s neighbors.\footnote{In this case it does not matter whether the sum is limited to span over neighbors or all the vertices because $P_{ij} = 0$ for $S_j$ that are not neighbors.  This convention will be used throughout - including cases where limiting an operation to neighbors matters.}\end{define}
Because by definition \ref{Probability} the sum of the probabilities sum to one, state criticality will also be a scale-free measure with values between zero and one.  All these criticality measures quantify the constriction of future possibilities on a state-by-state basis which is useful if we want ``to keep our options open".  As we will see later that is sometimes exactly what we want to measure, but sometimes we will want to measure system dynamics with reference to some particular features and that is what the following definitions for tipping points allow us to measure.

\subsubsection{Critical Levers}

As a refinement of levers from definition \ref{Levers} we can apply the lever concept to critical states to identify another feature of system dynamics.
\begin{define}\label{CriticalLever} A state has a \textit{critical lever} if a change in that aspect (or those aspects) of the state will reduce the reach. \end{define}  
This merely combines the concept of a lever with the concept of critical behavior (definition \ref{CriticalBehavior}). By looking more deeply at the aspects driving the state changes and calculating the magnitude and strength of different critical levers we can gain a better understanding of how microfeatures generate the macrobehavior of the model.

\subsection{Tipping Behavior}

As mentioned in the section introduction, the common feature of the measures in this section is that some states or transitions mark a shift in the properties of a system's dynamics.  For the criticality measures above the difference was the number of reachable states.  The following measures generalize to any sets distinguished by a chosen characteristic. Given states exhaustively compartmentalized by the property (or properties) of interest the following techniques can find where shifts occur and measure their magnitude.

\subsubsection{Tipping Points}

For some models we are interested in the achievement of a particular state (e.g. an equilibrium) or a particular system behavior (e.g. a path linking two states).  We denote the particular state (or set) of interest as the \textit{reference state} (or \textit{reference set}).  Below we will see examples of specific reference states (e.g. attractors and functional states) but first the general case.  There are many ways in which behavior may change with respect to a reference state or set (e.g. probability of reaching it, probability of returning to it, or probability of visiting an intermediate state): each property may partition the states into different equivalence classes (groups with the same value of the property).  It is the movement between equivalence classes that counts as tipping behavior.
\begin{define}A \textit{tipping point} is a state which is in the perimeter of an equivalence class for some property. 
\end{define}
Recall from definition \ref{Perimeter} that perimeter states are those from which the system's dynamics can leave the specified set.  Because the sets here are determined by the properties of system behavior leaving a set implies a change in that behavioral property - and that is a tip.  This definition does not preclude that the system could tip back into a previously visited set: that possibility depends on what property is establishing the equivalence classes.  

\begin{exa}\label{ClimateChange1} A climate change model that relates the CO$_{\textrm{2}}$ content of the atmosphere to global temperature may have states that are grouped together according to a shared property of those states (e.g. sea level, precipitation, glacial coverage).  Due to feedback mechanisms in the system it is likely the case that these qualitative features change in punctuated equilibria (see also example \ref{PunctuatedEquilibria}) thus producing equivalence classes for some states of the system.  \textit{Ex hypothesi} people can manipulate the level of CO$_{\textrm{2}}$ to higher or lower values.  The values at which the property shifts happen may differ for the increasing and decreasing directions, but the point is that CO$_{\textrm{2}}$ levels could raise temperatures to the point where glaciers disappear and then later lower past the point where glaciers will form again.  For some systems behavior can tip out of a equivalence class and then later tip back in.  Phase transitions in condensed matter physics are another example of reversible tipping behavior.  So while some have posited that tipping points are points of no return for system behavior, that turns out to be true only for certain systems and is not properly part of the definition.
\end{exa}

Dynamics of staying, leaving, returning, and avoiding a specified set of states will be covered in the section below on robustness.  Here we continue with ways to quantify changes in what is possible for system dynamics for different states and transitions.  These measures apply for any reference state or reference set, but for convenience and intuition pumping the following presentation will adopt the notation of attractor ($A_i$) for a reference state and set and $\mathbf{A}$ for a collection of reference states.  Recall from definition \ref{Attractor} that attractors may be equilibria or orbits and both possibilities were symbolized with $A_i$ and treated as singular.  That convention will be continued here. $\mathbf{A}$ is a collection of independent reference states and sets each of which satisfies a property while $A_i$ may be a set of states that collectively satisfies a particular dynamical property (such as an orbit as a whole satisfies the dynamical property of an equilibrium).

\begin{define}\label{EnergyLevel} The \textit{energy level} of a state is the number of reference states within its reach.  We write this as $E(S_i)$ and it equals \[ | \bigcup_{A_i}^{\mathbf{A}} A_i \in \mathbf{R}(S_i)  | \] \end{define}
Energy level quantities partition the system's states into equivalence classes. 
\begin{define}\label{EnergyPlateau} The equivalence class mapping created by states' energy levels is called the system's \textit{energy plateaus}.  Each energy plateau is a set \[ \bigcup_{i,j}^{N} E(S_i) = E(S_j) \] \end{define}.
\begin{define}\label{EnergyPrecipice} The change in energy across a transition is called an \textit{energy precipice} or \textit{energy drop}. We can measure the magnitude of an energy precipice in the obvious way: \[ \triangle E(\overrightarrow{S_i \: S_j}) = E(S_i) - E(S_j)\] \end{define}
\begin{figure}[!ht]
\centering \caption{Energy Plateaus}
\label{EnergyLevelsFigure}
\begin{center}
\includegraphics{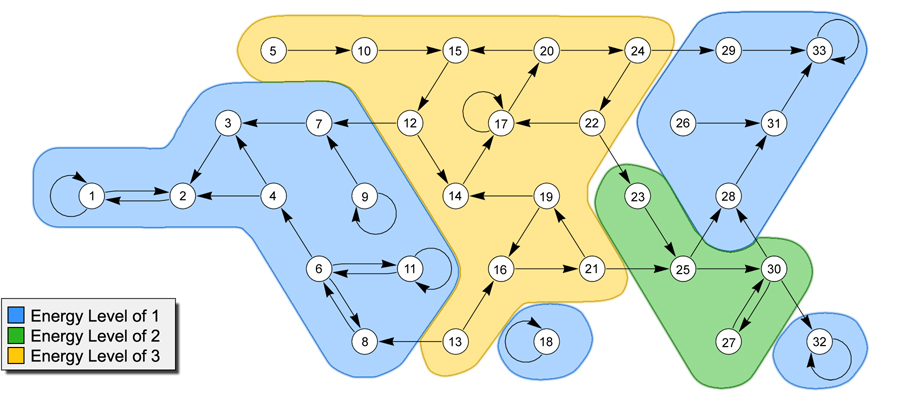}
\end{center}
\end{figure}
\begin{thm}\label{EnergyDrop} An energy precipice is never negative: $\triangle E(\overrightarrow{S_i \: S_j}) \geq 0$. \end{thm} 
\begin{proof} From definition \ref{EnergyLevel} the energy level of any state $S_i$ is the number of attractors in $\mathbf{R}(S_i)$.  By theorem \ref{DecreasingReach} for any $\overrightarrow{S_i \: S_j} \: \mathbf{R}(S_j) \subseteq \mathbf{R}(S_i)$.  This decreasing reach property implies \[ | \bigcup_{A_i}^{\mathbf{A}} A_i \in \mathbf{R}(S_j) | \leq | \bigcup_{A_i}^{\mathbf{A}} A_i \in \mathbf{R}(S_i)  |  \Rightarrow E(S_j) \leq E(S_i) \]
Since the end state of a transition always has a lower or equal energy level, $E(S_i) - E(S_j)$ is always greater than or equal to zero. \end{proof}

We now can measure the degree to which a state is likely to be the site of a tip (in a similar fashion to definition \ref{StateCriticality} of state criticality above).  
\begin{define}\label{Tippiness} The \textit{tippiness} of a state is the probabilistically weighted proportional drops in energy of its immediate successors:  
\[ 1 - \sum_{j=1}^{k} P_{ij} \left(\frac{E(S_j)}{E(S_i)} \right) \] \end{define}
\begin{thm}\label{TippinessRange} Tippiness ranges from zero to one.\end{thm}
\begin{proof} The lower bound occurs when all neighbors can reach the same number of reference states: $\forall k \: E(S_i) = E(S_k)$.\footnote{Recall the convention that $k$ denotes the degree of a state and $S_k$ is a neighboring state.} In this case tippiness equals $1 - \sum_{j=1}^{k} P_{ij} \cdot 1 = 0$ by theorem \ref{Probability}. By theorem \ref{EnergyDrop} $E(S_i) \nless E(S_k)$. When $E(S_i) > E(S_k)$ the upper bound occurs when a state has every attractor and only attractors as neighbors. Again by theorem \ref{Probability} $\sum_{j=1}^{k} P_{ik} = 1$.  Attractors by definition \ref{Attractor} have an energy level of one and the energy level of $S_i$ in this case is $|\mathbf{A}|$ so $S_i$'s tippiness is $1 - \sum_{j=1}^{k} \frac{P_{ij}}{|\mathbf{A}|} = 1 - \frac{1}{|\mathbf{A}|}$. Thus the upper bound of $S_i$'s tippiness goes to 1 as $|\mathbf{A}| \rightarrow \infty$. \end{proof}
Note that tippiness uses the ratio of energies rather than the difference; this makes tippiness a dimensionless metric and thus comparable across any state or system.  Sometimes one will be more interested in minimizing the magnitude of energy drops, or avoiding states with the highest expected magnitude of energy drops, which have obvious formulations given the above definitions.

\subsubsection{Tipping Levers}

As another refinement of levers from definition \ref{Levers} we can apply the lever concept to tipping points to relate tips to the individual aspect changes that drive them.
\begin{define}\label{TippingLever} A state has a \textit{tipping lever} with respect to some specified set of states if a change in that aspect (or those aspects) of the state will take the system out of that set of states. \end{define}  
This merely combines the concept of a lever with the general concept of tipping behavior.  So while every state accept equilibria has levers, only perimeter states have critical levers.  Identifying the tipping levers of certain sets of states is precisely what we'd like a ``tipping point'' analysis to reveal because it is just the aspects of a tipping point that actually change when the system tips out of a set of states.  These ideas are further refined in the analysis of robustness and related concepts below.

\section{Robustness-Related Measures}
This section will use the Markov model framework provided above to establish formal definitions of several related concepts: robust, sustainable, resilient, recoverable, stable, and static; as well as their counterparts: susceptible, vulnerable, fragile, and collapsible.  As before, it is unlikely that any mathematically precise definition will maintain all the nuances of the full concept sharing the same name.  Furthermore, existing definitions and formal treatments of the same concepts or that use the same terms risk distracting from or confusing the thrusts of this research.  What is more important than the terms used is that the definition is useful and we can easily refer to it, though some care has been spent on finding the closest word to the provided definition.  

Adding to the ambiguous and often synonymous usage of these terms are the conceptual questions that arise in considering their dispositional nature.  Dispositional properties are philosophically troubling for many reasons but the philosophical troubles will not interfere with their definition and ascription here (mostly these stem from the question of whether their subjunctive conditional status distinguishes them from categorical properties).  I mention this here because I want to hint that in addition to the obvious and direct application of this methodology to improve the performance capabilities of systems, it may also produce insights into the nature of dispositional properties in general.  This analysis reveals how these particular dispositional properties are behavioral in nature and emerge from the microbehavior of the system components.  The philosophical issues will be addressed in separate work (see future work subsection below), but it may be interesting to the reader to consider how to apply the following methodology to investigate other dispositional properties such as soluble, malleable, affordable, and differentiate them from other types of properties with the much sought after necessary and sufficient conditions.

\subsection{Stable, Static, and Turbulent}

The measures in this first set are conceptually simple with intuitive mathematical definitions and straightforward algorithms.  They nevertheless identify important features of system dynamics and act as building blocks for more sophisticated measures.  
\begin{define}\label{StateStability} A \textbf{state's} \textit{stability} is how likely that state is to self-transition.  $S_i$'s stability is \[ P(s_{t+1} = S_i | s_t = S_i).  \] \end{define}
While this may seems a trivial property, it is consistent with a useful distinction from system dynamics: the difference between stable and unstable equilibria.  Due to the resolution of the Markov model's states, an attractor state will include a neighborhood of aspect values around the equilibrium point values.  Thus exit behavior from the attractor node includes the response to small perturbations to (or variations around) the equilibrium point values.  Stable states will tend to stay within this neighborhood and this is reflected in a high self transition probability value.  Since values that are nearby an unstable equilibrium but not exactly on the equilibrium point values will tend to move away from the equilibrium values, we would see this reflected in low self-transition probabilities.  These results exactly match attributions of stability and instability in the Markov model via definition \ref{StateStability}

We can extend stability to apply to sets of states in the obvious way.  \begin{define}\label{SetStability} The \textit{stability} of a \textbf{set} is the probability that the system will not transition out of the set given that the system starts within the set.  We calculate this as the average of the individual states' exit probabilities, so set stability is \[ \frac{1}{|\mathbf{S}|} \sum_{S_i \in \mathbf{S}} P(s_{t+1} \in \mathbf{S} | s_{t} = S_i). \] \end{define} 
This is a crude measure because is doesn't properly reflect the probability of staying within the set over time, it only looks one time period ahead.  A more sophisticated notion of staying within a set of states is presented by definition \ref{Sustainability} below of sustainability, but in some cases (discussed in that subsection) the two measures generate the same value.  

The word `static' is often used to indicate a lack of dynamics in a system, and that is the sense attached to the following formal definition.  Recalling that self-transitions can be interpreted as a lack of transition it aggregates the lack of transitions among states.
\begin{define}\label{Static} The degree to which a set is \textit{static} is the average of the states' stability values: \[\frac{1}{|\mathbf{S}|} \sum_{S_i \in \mathbf{S}} P(s_{t+1} = S_i | s_t = S_i)  \] \end{define}
This definition, though simple, captures how likely a system is to be in the same state for consecutive time steps in a way that is comparable across sets and systems with different numbers of states.  What this definition fails to capture is that sets with equilibria will spend an infinite amount of time in them whereas sets lacking equilibria will continue to transition for eternity (even if that's within an attractor); and yet because this measure uses average stability it is easy to construct cases where an equilibrating systems has a lower static level on the given definition.  

Static and stable set measurements are similar in their calculation but distinct in their sense.  Set stability is a measurement of lack of change, but it is a lack of change out of a set (though it ignores dynamics that stay with in the set).  It is therefore only applicable if the set chosen is smaller than the whole system.  Staticness can apply to the whole system and is useful for comparing systems' overall level of dynamism.  
\begin{thm}If $\mathbf{S} = S_i$ (i.e. a set with one state) then the static measurement equals the stability measurement.
\end{thm} This theorem clearly follows from the fact that the sum of transitions staying within a set equals the sum of self-transitions for a set of one state.  Also in this case set stability naturally equals state stability because the set is a state.
\begin{thm}For any set, the set stability measure is always greater than or equal to the level of staticness. \end{thm} \begin{proof} For any given $S_i \in \mathbf{S}$, the set stability measure is \[ P(s_{t+1} \in \mathbf{S} | s_t = S_i) = \sum_{j=1}^{k} P(\overrightarrow{S_i \: S_j} | S_j \in \mathbf{S}) + P(\overrightarrow{S_i \: S_i}) \] The probability of transferring to another state in $\mathbf{S}$, $\sum_{j=1}^{k} P(\overrightarrow{S_i \: S_j} | S_j \in \mathbf{S})$  is greater than or equal to zero.  If that equals zero then $S_i$'s contribution to set stability becomes $P(\overrightarrow{S_i \: S_i})$ which is equal to $S_i$'s contribution to the static measure.  If $P(\overrightarrow{S_i \: S_j} | S_j \in \mathbf{S}) > 0$ for any $S_i \in  \mathbf{S}$ then $\mathbf{S}$'s set stability is greater than its degree of being static.\end{proof}

On the other end of the spectrum as stability and staticness are measures of how likely the system is to change states.  Because the above measures are defined in terms of probabilities, most simple measures of the presence of dynamics can be calculated as one minus the appropriate measure above.  There is one additional simple measure to present here; it is a rough measure of how predictable state changes are.
\begin{define} The \textit{turbulence} of a set is the average percentage of states that its states can transition into.  We can calculate $\mathbf{S}$'s turbulence with the average ratio of each state's degree to the number of states in $\mathbf{S}$: 
 \[ \frac{1}{|\mathbf{S}|} \sum_{S_i \in \mathbf{S}} \frac{k}{|\mathbf{S}|} \] \end{define}
This measure ranges from zero to one where the zero case occurs as $|\mathbf{S}| \rightarrow \infty$ and a turbulence of one means that the set is fully connected (including self-transitions for each state).  The idea is that when each state has only a few possible transitions then there are far fewer possible paths through the system dynamics.  If each state can transition into many others then, like with the common usage of `turbulent', there is a great deal of uncertainty regarding the path that a series of transitions will take.  Plotting the degree distribution would reveal a set's (which might be the whole system or just a portion) turbulence profile.  If one were to find something like a power-law distribution (where a few states have many transitions and most have just a few transitions) the high-degree states would seem to satisfy yet another concept of tipping behavior.  Combining turbulence profiling with (for example) the identification of perimeter states could be used to classify systems by the dynamical properties (see future work section for more details).  

Though the turbulence measure may provide sufficient information in many systems, it fails to differentiate the effects of high and low probability transitions.  Transition weights clearly play a role in determine how confident one can be that a particular trajectory will be taken rather than another.  For example, if all but one of a state's transitions have very small probabilities associated with them then the set should be considerable less turbulent than if all the transitions are equally probable.  
\begin{define}\label{WeightedTurbulence} As a refinement of turbulence, \textit{weighted turbulence} of the state $S_i$ equals zero if $k=1$ and for $k>1$ can be calculated as 
\[  \sum_{j=1}^{k} 1 - (P(\overrightarrow{S_i \: S_j}) - \frac{1}{k})^{2} \] \end{define}
Because by definition \ref{Probability} the sum of exit probabilities sum to one, the average of the exit probabilities is $\frac{1}{k}$ regardless of the number and their individual weights.  The innermost component of this calculation, therefore, finds the difference between each transition and the average weight and squares it.  Squaring has the dual effect of producing an absolute value and intensifying differences; the intensification is not crucial and is merely adopted by convention.  Because turbulence is maximal when each weight is equal, we subtract the differences from one to calculate each state's turbulence.\footnote{The reason for the separation of the value for the $k=1$ case is that the definition produces a value of one instead of zero.  This is an artifact of the fact that if there is only one edge then all the edges have a weight equal to the average weight - and that is the case that produces maximal turbulence for all other $k$s.}  To determine the weighted turbulence of a set we simply average each included state's weighted turbulence

\subsection{Sustainable and Susceptible}

In common parlance something is sustainable if it can perpetually maintain its operation, function, or existence. It is often used in connection to environmental considerations such as whether humans are using up resources faster than they can be replenished or to the ecological question of whether population dynamics will drive any species to extinction.  Political institutions, academic reading groups, pools of workers, and any other system that undergoes inflows and outflows of its parts and might collapse or fail is a potential subject for sustainability considerations.  

Roughly speaking, for this analysis a set of states is sustainable if the system can stay within that set of states.  There are multiple ways to calculate a measure of this sort and each reports a slightly different concept of sustainability.  As a crude approximation to the long-term sustainability we can find the cumulative sum of the $t^{\mathrm{th}}$ power of the set stability measure from definition \ref{SetStability} up to some sufficiently large $T$ (see observation \ref{BigEnough} below).  We could call such a measure \textit{naive sustainability} and identify conditions for its appropriate application, but instead we will move on to a more sophisticated measure.

The previous measure is crude because within a set there may be (for example) heavily weighted cycles such that if the system starts in one of the cycle-states it is very likely to go around the cycle for a long period.  To properly account for this, while still remaining agnostic over which state of $\mathbf{S}$ the system starts in, we calculate a refined sustainability measure. 
\begin{define}\label{Sustainability} The \textit{sustainability} of $\mathbf{S}$ is the average cumulative long-term probability density of future states that remain in the set starting from each state in the set.\footnote{The algorithm for calculating sustainability starts with a vector of ones for the states in the set and zeros for states out of the set.  It then iteratively applies the this vector to the whole adjacently matrix taken to the $t^{\mathrm{th}}$ power up to $t = T$ (see observation \ref{BigEnough}).  This is done each time to clear out the probability mass outside the set so that it can't return.  The result is the sum of the  resulting vectors for the states in $\mathbf{S}$ each step.    Usually algorithms are been banished to the appendix, but understanding this calculation is likely to make understanding definition \ref{Susceptible} of susceptibility much clearer.} \[ \frac{1}{|\mathbf{S}|} \sum_{S_i \in \mathbf{S}} \sum^{T}_{t=0} P(s_{t+1} \in \mathbf{S} | s_{t} \in \mathbf{S}) \] \end{define}
\begin{obs}\label{BigEnough} If the chosen set does not contain an attractor then this calculation is unproblematic because some probability density ``escapes'' the set each iteration and there exists some time $T$ after which the remaining probability density in each state of $\mathbf{S}$ is less than any arbitrarily chosen minimum resolution.\end{obs}
Sustainability measurements are therefore only appropriate for sets that do not include attractors.   If an application to sets including attractors were deemed useful then we could separate out the basin(s) of the attractor(s) from the other states and apply the sustainability measure above to the remaining states of the set.  It is still not clear how to recombine the two subsets into a single measure or how to cope with epistemic barriers to knowing whether a set contains an attractor before running the analysis and thus whether this would be necessary.  Because sustainability is a cumulative measure the calculation will not produce a probability. But because the cumulative sum is divided by the size of $\mathbf{S}$ it is normalized and comparable across differently sized sets. What we uncover through this process is a measure of the expectation over time that the system will stay in the set given that the system starts somewhere in it.  
\begin{figure}[!ht]
\centering \caption{Sustainability Measures for Two Energy Plateaus}
\label{SustainabilityFigure}
\begin{center}
\includegraphics{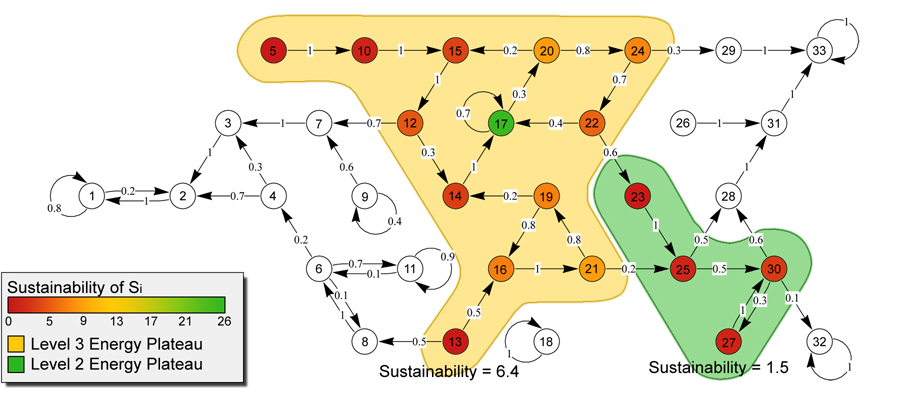}
\end{center}
\end{figure}

\begin{exa} If $\mathbf{S}$ is chosen to be an energy plateau then all exit transitions are one-way.  If there are no cycles in $\mathbf{S}$ then the probability mass will quickly dissipate from within $\mathbf{S}$ yielding a very low sustainability measure.  If there is at least one cycle then there is a chance that the system will stay within the set indefinitely, but since that cycle cannot be an attractor the system will leave the set in expectation.  The stronger the weights of the transitions among the cycle states the greater the sustainability measure.  

The energy plateau application is especially helpful for the ``keep our options open'' mindset and where each equilibria is a different form of system failure (e.g. for dissipative structures).   It can even be useful to calculate the sustainability of a basin of attraction (excluding the attractor itself) - which is also an energy plateau.  The system may exhibit interesting and long-lived behavior within a basin of attraction that may reveal much more about the processes affecting a system than just which attractor it is likely to end up in.  Time to equilibria may be on the galactic time-scale and knowing that will alter our interpretation of the system's characteristics.
\end{exa}

The term `susceptible' is typically followed by `to' and an indication of what the thing is susceptible to.  I preserve that usage with the measure presented here.  We will talk of sets being susceptible and sets can be defined according to different properties for different applications.  States are what sets are susceptible to.  This may sound odd, but the probability of transitioning out of a set depends on which state within the set the system is currently in.  So the characteristic represented by staying within a set risks being lost to a degree contingent upon the state.
\begin{define}\label{Susceptible} The degree to which $\mathbf{S}$ is susceptible to $S_i$ is how much more (or less) likely it is to transition out of $\mathbf{S}$ conditional on it being in a particular state $S_i$ of $\mathbf{S}$ compared to the sustainability of $\mathbf{S}$ overall.  
\[\sum_{t=1}^{T} P(s_{t+1} \in \mathbf{S} | s_{t} \in \mathbf{S} \textrm{ and } s_{0} = S_i) - \textrm{sustainability of }\mathbf{S} \] \end{define}
Given this definition we can see that a positive susceptibility means a lower probability to stay within $\mathbf{S}$.  

We can also determine and measure the set's susceptibility to the lever points of an aspect.  Recall from definition \ref{LeverPoint} that the lever points of an aspect are all those transitions that result from a change in that aspect.  We can calculate the susceptibility of $\mathbf{S}$ to a collection of lever points as
\[\sum_{t=1}^{T} P(s_{t+1} \in \mathbf{S} | s_{t} \in \mathbf{S} \textrm{ and } \overrightarrow{s_{0} \: s_{1}} \in \bigcup_{\overrightarrow{S_i \: S_j}}^{\mathbf{E}} X_{h(i)} \neq X_{h(j)}). \] 

\begin{exa} Sustainable/susceptibility analysis can be used to help systems maintain a performance level.  We can take the set to be a contiguous collection of states that count as functional in some system: such as all the configurations of an airplane that the autopilot can manage.  For the airplane system some state changes will be exogenous perturbations due to environmental factors (wind, rain, pressure, lightening, passenger movement, etc.), others will be endogenous control adjustments by either the pilot or the autopilot, and some will be a mix.  First one would calculate the sustainability of the whole set of autopilot capable states. Then one would calculate how susceptibility that set is to each state (or smaller collection of states).  Using this information the autopilot and/or pilot could select actions that minimize susceptibility across the states visited and this means maximizing the probability of staying within the set of autopilot capable states.  This example can be generalized to any case where maintaining functionality is the modeler's goal.
\end{exa}

\subsection{Resilient, Fragile, and Recoverable}

Stability, Staticness, and Sustainability are different ways to measure a system's dynamics tendency not to leave a state or set; we now turn to measures of returning to a state or set once it has been left.  
\begin{define}\label{Resilience} A state's \textit{resilience} is the cumulative probability of returning to a state given that the system starts in that state. The resilience of $S_i$ equals
\[\sum_{t=1}^{T} P (s_{t} = S_i | s_0 = S_i) \] \end{define}
It is the sum of the individual probabilities of returning in $1,2,3, \ldots, T$ time steps.  Because the sum of exit probabilities of every state equals one and the probability of traversing a path is the product of the states along the path this cumulative sum is always less than or equal to one and is a true probability measure.  

\begin{define}\label{Fragility} A state's \textit{fragility} is a measure of how likely it is that the system will never return to that state.  This is just one minus the resilience of that state.\end{define}
So equilibria have zero fragility and states with no return paths have a fragility of one. Measuring the degree of fragility requires the same calculation as measuring the resilience, but finding out whether a state is ever revisited is much easier because we can utilize our definition of a state's reach.
\begin{define}\label{Brittle} A state is \textit{brittle} if and only if it has a fragility value of one (i.e. a resilience value of zero).  Brittle states are the ones such that \[ S_i \not \in \mathbf{R}(S_i). \] \end{define}
Except for the brittle states which have a specific formal significance, the choice of whether to use a resilience or fragility measure will depend on which feature the user would like to highlight (glass half-full or half-empty).  
\begin{figure}[!ht]
\centering \caption{A System's Brittle States}
\label{BrittleFigure}
\begin{center}
\includegraphics{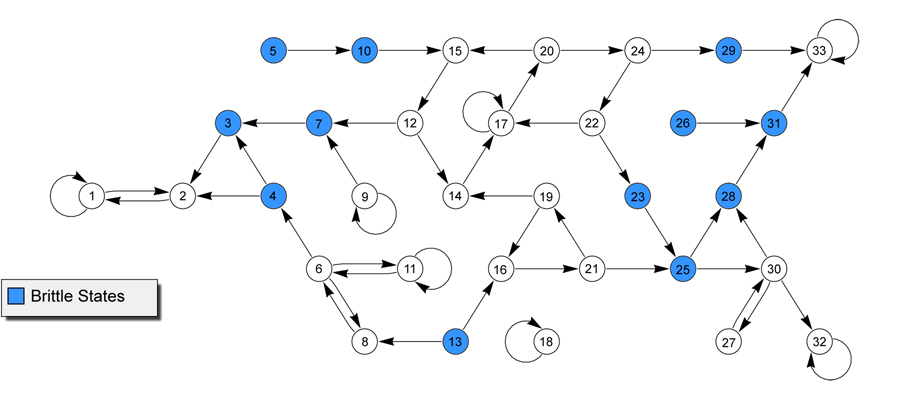}
\end{center}
\end{figure}
We can also define the resilience and fragility of a set in an analogous way. 
\begin{define}\label{SetResilience} \textit{Set resilience} is the probability that the system will return to a set if the initial state of a sequence is within the state.  
\[ \sum_{t=1}^{T} P (s_{t} \in \mathbf{S} | s_{0} \in \mathbf{S}) \] \end{define}
Though the definition is exactly parallel to the single-state case, the algorithm to calculate this probability is considerably more difficult.  
\begin{obs} A few facts about entering and leaving sets will help refine our understanding.
\begin{itemize}
\item[(i)] Transitions exit $\mathbf{S}$ through the perimeter states $\mathbf{S}_{out}$ of  $\mathbf{S}$.
\item[(ii)] Transitions enter $\mathbf{S}$ through a set of entry points $\mathbf{S}_{in}$ of $\mathbf{S}$.
\item[(iii)] We can refine the definition of set resilience to $\sum_{t=1}^{T} P (s_{t} \in \mathbf{S}_{in} | s_0 \in \mathbf{S}_{out})$.
\end{itemize} \end{obs}
So to calculate set resilience we need first to find all the paths from each element in $\mathbf{S}_{out}$ to each element in $\mathbf{S}_{in}$.  In the worst case this can be done in $O(|\mathbf{E}|+|\mathbf{N}|)$ time via a breadth-first search.  Set fragility is one minus set resilience.  

\begin{thm}\label{ResilienceIsZero} If $\mathbf{S}$ is an energy plateau then the resilience of $\mathbf{S}$ is zero. \end{thm} \begin{proof} By definition \ref{EnergyPlateau} an energy plateau contains all the states in the system with the same number of attractors in their reach.  Any transition out of such a set would be to a state with a different energy level and by theorem \ref{EnergyDrop} it must be a lower energy level.  Also by theorem \ref{EnergyDrop} no transition can be to a higher energy level.  Hence if a system transitions out of an energy plateau then it can never transition back into it.  If the system cannot transition back into the set $\mathbf{S}$ then by definition \ref{SetResilience} $\mathbf{S}$'s resilience is zero. \end{proof}

The susceptibility measure determines how sustainability changes depending on the specific starting state.  We may wish to have a similar measure for resilience that reports how likely the system dynamics are to return to a state given that it exits via a particular transition.  
\begin{define}\label{Recoverable} A transition out of the set is \textit{recoverable} to the degree that the system will return to the set after the transition.  $\mathbf{S}$ is recoverable from $\overrightarrow{S_i \: S_j}$ to the degree calculated by \[ \sum_{t=1}^{T} P (s_{t} \in \mathbf{S} | \overrightarrow{S_i \: S_j} \textrm{ and } S_{i} \in \mathbf{S} \textrm{ and } S_{j} \not \in \mathbf{S} ) \] \end{define}
Note that leaving via a particular transition is the same as exiting due to a particular lever change.  Thus we can uncover the recoverability of a set of lever points (from definition \ref{LeverPoint}) for a particular aspect as the average of the recoverability of each transition in it.  Also note that there may be multiple paths from $S_i$ back into each $\mathbf{S}_{in}$ of $\mathbf{S}$.  Each path leading from $\mathbf{S}$ back into $\mathbf{S}$ can be called a \textit{recovery path}.  
\begin{exa} Continuing with the autopilot example, imagine that there are many known points of failure for maintaining autopilot control.  Each of these is a transition out of the set $\mathbf{S}$ via a known lever change.  But not all failures are equally as problematic.  By calculating the recoverability of each of the failure transitions they can be ranked by their seriousness.  Such a ranking can guide both the pilot in adjusting to the failure and the autopilot to avoid it in the first place.  Again, the autopilot example can be generalized to the maintenance of any system: political regimes, sports clubs, ecosystems, viable crop production, etc. \end{exa}

\subsection{Reliable, Robust, and Vulnerable}

Sustainability measures the likelihood of a system's dynamic's staying in a certain set given that it starts within that set and resilience measures how likely it is to return to the set if the dynamics leave the set, but these measures do not include the case where the system's state starts outside the set and then enters it.   When the set of interest is an energy plateau resilience is always zero (as shown by theorem \ref{ResilienceIsZero}) but the set may still receive probability mass from parts of the system with higher energy levels.  And in cases where a non-equilibrium analysis is appropriate we might be comparing different subsets within an energy plateau (e.g. the relative probability mass of two cores within the mantle of an energy plateau - see example \ref{PunctuatedEquilibria} below).  In this final subsection the above measures will culminate in the most inclusive measures of system robustness.

Before defining the measure that allows for inflow, we first define a measure that combines the features of sustainability and reliability.  
\begin{define}\label{Reliable} The \textit{reliability} of a set is the average cumulative long-term probability density over the states in the set given that the system starts within that set.  \[ \frac{1}{|\mathbf{S}|} \sum_{S_i \in \mathbf{S}} \sum^{T}_{t=1} P(s_{t} \in \mathbf{S} | s_{0} \in \mathbf{S}) \] \end{define}
This measure combines the concepts of sustainability and resilience, but it is not just the sum of those two measures.  Reliability starts the flow in the set and calculates the probability of being in each state on each consecutive time step.  It does restrict the probability mass summation to the specified set, but it tracks probability mass throughout the system.  The reason that this isn't merely a sum of resilience and sustainability is because when combining those two it was not possible to track probability mass that leaves the set, cycles back into the set, and then circulates within the set (and maybe even repeats this process).  With reliability we can reincorporate probability flow that leaves and then re-enters the set.  A characteristic captured by the chosen set is reliable if it can be maintained or, if lost, can be regained.  

\begin{thm} If $\mathbf{S}$ is an energy plateau then $\mathbf{S}$'s reliability equals $\mathbf{S}$'s sustainability. \end{thm} 
This theorem does not follow directly from theorem \ref{ResilienceIsZero} because there is no direct link between resilience and reliability, but the reasoning is the same.  Because there cannot be any paths leading out of an energy plateau back into it, all the probability mass that $\mathbf{S}$ gets for the reliability measure is from the initial distribution.  Leaving mass never returns so that produces an equivalent measure as not counting the returning mass: this is the sustainability measure.  Hence the sustainability values in figure \ref{SustainabilityFigure} are also those energy plateaus' reliability values.
 
Finally we add to the reliability measure the possibility that the system did not start in the set, but transitions into it.  
\begin{define}\label{Robust} The \textit{robustness} of a set is the average cumulative long-term probability density over the states in the set given that the system may start at any state.  \[ \frac{1}{|\mathbf{S}|} \sum_{S_i \in \mathbf{S}} \sum^{T}_{t=0} P(s_{t} \in \mathbf{S}) \] \end{define}
Robust characteristics not only have high retaining power and recoverability, they also draw the system in from states outside the characteristic set.  Sets with high robustness values are sets that the system's dynamics tends towards.  That description makes robust sets sound a lot like attractors; and this is as we would expect.  Attractors will typically have high robustness measures on account of their perfect sustainability and the fact that typically several states will lead into them.\footnote{Recall that sustainability is not intended to apply to sets with attractors as parts of them.  The measure is, however, unproblematic if the set \textbf{is} an attractor (which yields a sustainability of $T$) or a proper subset of an attractor.}  The attractor-like behavior related to robust sets provides interesting and useful insights into many systems' dynamics.  

\begin{exa}\label{PunctuatedEquilibria} Sets that behave like (and are defined as) equilibria in other modeling techniques may be revealed to be highly robust sets under the current analysis.  The phenomena of \textit{punctuated equilibria} describes a system that spends long periods of time in characteristic patterns with interspersed and short-lived periods of rapid change.  In the Markov model representation we might see a mantle of an energy plateau with multiple highly robust cores.  These cores could have relatively short transition paths among them.  Each is a different cohesive pattern with larger probabilities of staying in than going out.  But because these cores are not attractors the system will eventually transition out of them and into the next core.  \end{exa}
 
\begin{exa}\label{DissipativeStructure} One of the foci of complex systems science is the study of the self-maintaining (or \textit{autopoietic}) nature of dissipative structures.  Dissipative structures are those where a continual flow of energy, matter, or other resource is necessary to maintain system structure and performance.  Biological systems are like this, constantly changing and adapting to maintain functionality, and so are many other complex systems.  These are systems where there are no equilibria\footnote{There must be at least one attractor per system, but that attractor may be an orbit consisting of every state in the system.} or all equilibria are states to be avoided so that the energy level of the system remains mostly constant.  Some set(s) of states are preferred to others for exogenous reasons (functionality, performance, diversity, longevity, or other utility measures) and the goal is to maximize time spent in the desired states.  The goal might also be to maintain some characteristic feature of transient system behavior. The current techniques offer new measures of behavior for non-equilibrium analysis. These can be used to embed an existing equilibria model into a larger context and/or to push down the level of analysis to see what is happening inside an ``equilibrium'' state.\end{exa}

Using a definition parallel to that of susceptibility, the following measure calculates how much more (or less) likely the system is to be in the set conditional on the dynamics starting in a particular state (not necessarily in that set).  
\begin{define}\label{Vulnerable} A set's \textit{vulnerability} is the difference in the average long-term probability density over the states in the set compared to the density generated by starting in $S_i$  \[ \left( \sum_{t=1}^{T} P(s_{t} \in \mathbf{S} | s_{0} = S_i) \right) - \textrm{robustness of }\mathbf{S} \] \end{define}

\section{Path Sensitivity in System Dynamics}

This section's status is somewhere between completed work and future work - which I suppose makes it work in progress.  It has not been prepared to the degree of rigor of the previous sections, but there is sufficient material here to provide a strong indicator of how a Markov model representation can be used to uncover a variety of measures related to path sensitivity.  The algorithms for the presented measures (and more) have been prepared, but the mathematical and conceptual work needs more smoothing and filling in.  It is offered in its present form for review, evaluation, and feedback purposes.

Path sensitivity is not a single, well understood and properly defined concept.\footnote{I have named the more general idea that how systems behave changes based on their past `path sensitivity' and have reserved the name `path dependent' for the specific form of path sensitivity highlighted in the works of Page and Kollman \& Jackson.}  For starters, different features of system dynamics can have the path sensitivity property: Outcomes can be path sensitive, processes can be path sensitive, measures can be path sensitive, and paths can be path sensitive.  Previous work has focused on explaining why a particular system exhibits path sensitivity or how a particular mechanism can generate it, but they have not provided general, causally agnostic measures applicable to any system's dynamics.  What follows is a collection of distinctions and formal definitions for several different forms of path sensitivity.  And like the above measures of tipping- and robustness-related concepts, these measures apply to Markov models representing either observational or model-generated data.  

\subsection{Path Preclusion}

One type of path sensitivity results when a transition excludes some states from any possible future of the system's dynamics.  Such an occurrence is closely related to the measures of criticality and tipping behavior presented in definitions \ref{Criticality} and \ref{Tippiness} respectively.   
\begin{define}\label{WeakPathPreclusion} Any reduction in the size of the reach across a transition is instance of \textit{weak path preclusion}.  The degree of path preclusion of $\overrightarrow{S_i \: S_j}$ is the criticality measure of $\overrightarrow{S_i \: S_j}$.
\end{define}
Such a definition is likely to be useful in limited number of models but it is conceptually intuitive that merely excluding a set of states from the future of a system's dynamics is a relevant form of path sensitivity.  In applications geared towards preserving, tracking, or monitoring some characteristic of the system or its behavior we need a concept that accounts for the preclusion of that characteristic.  
\begin{define}\label{StrongPathPreclusion} When there is a reduction in the number of reference states or sets (e.g. attractors or specific equivalence classes) that can be reached then this is tipping behavior that exhibits \textit{strong path preclusion}.  The strong path preclusion of $\overrightarrow{S_i \: S_j}$ is measured by the tippiness of $\overrightarrow{S_i \: S_j}$. \end{define}
Recall from example \ref{ClimateChange1} that not being able to return to a state or set was \textbf{not} properly part of the definition of tipping behavior.  Whether a tip is path preclusive or not marks an important feature of the transition.  In many cases we will be more interested in whether a transition is path preclusive than whether it is tipping out of some characteristic.   It is this dynamic and its importance for understanding system dynamics that gives tippiness its relevance as a measure.

\subsection{Trajectory Forcing}

Some systems have only one long-term outcome (a singular attractor) and some have none (when the whole system is an orbit).  In either case we might care less about where the system goes than how it gets there, i.e. which intermediate states the system realizes between two anchor states.  We might care about the exact path our dynamics takes through the system's states because some states are preferred to others (for exogenous reasons) or because they have different dynamical properties (e.g. susceptibility or vulnerability measures).  Two different tips may include the same reference states but force the system dynamics to take different paths to get there.
\begin{define}\label{TrajectoryForcing}\textit{Trajectory forcing} is when a particular transition sends the dynamics down a specified sequence of states. The \textit{force} of an exact path from $S_i$ to $S_j$ can be measured as the product of the probabilities of all the transitions required to stay that course: \[ \textrm{force of }\vec{\mathbf{S}}(S_i, \ldots, S_j) := \prod_{h=1}^{\ell} P \left( \overrightarrow{S_{h-1} \: S_h} | S_{h-1}, S_h \in \vec{\mathbf{S}}(S_i, \ldots, S_j) \right) \] \end{define}
Given this general definition it is clear that forcing is relative to the specific set of connected states.  If, for example, one wanted to maximize the probability of reaching a specific equilibrium then calculating the force of each path from the tips of the current core to that equilibrium would provide the necessary guidance (see figure \ref{TrajectoryForcingFigure}).  
\begin{figure}[!ht]
\centering \caption{The Force of Two Paths from $S_{19}$ to $S_{33}$}
\label{TrajectoryForcingFigure}
\begin{center}
\includegraphics{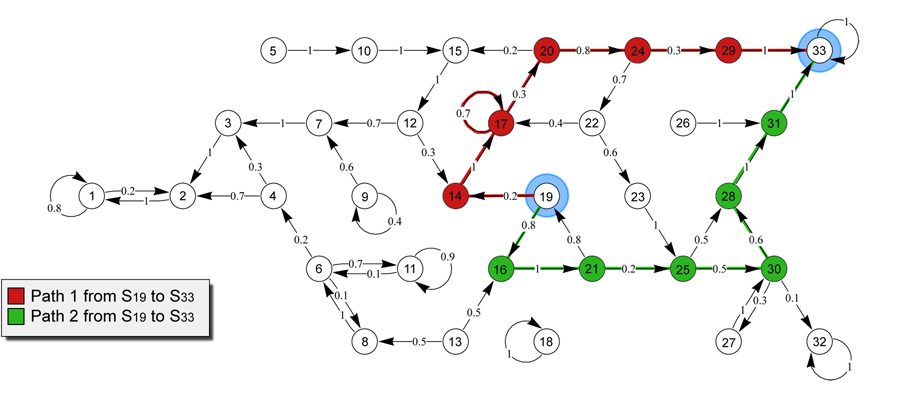}
\end{center}
\end{figure}
The red path from $S_{19}$ to $S_{33}$ in the figure has a force of $0.01008$ and the green path has a force of $0.048$.  Force is not the same thing as the probability of reaching $S_{33}$ given $\overrightarrow{S_{19} \: S_{14}}$ and $\overrightarrow{S_{19} \: S_{16}}$ respectively; that is calculated by the tippiness of those two transitions with only $S_{33}$ as a reference state.  These force measurements are the probabilities of following those exact paths to get between those two states.  For non-exact paths force can be calculated as the sum of the probabilities of the paths that visit each of the specified markers.  

\subsection{Path Dependence}
The term ``Markovian'' means memoryless in the related fields of mathematics where it refers to any processes wherein future probabilities of events do not depend on past occurrences (e.g. Poisson processes).  The Markov model representation may therefore seem an unlikely tool for uncovering dependencies in paths of system dynamics.  The conditional probability assignments to the transitions imply that the current state is sufficient for knowing the probability distributions of future states.  But another way to interpret this Markov model structure is that it encapsulates the probability distribution of future states \textbf{if} all one knows is the current state.  It can also be used to track correlations in the specific paths followed.
\begin{define}\label{PathDependence} A state's exit transitions are \textit{path dependent} if and only if the distribution of their probabilities changes conditional on previous states.  The degree to which $S_i$'s transitions are path dependent on a set of historical sets $\mathbf{S}_H$ equals \[ \sum_{j=1}^{k} \left( P(\overrightarrow{S_{i} \: S_{j}}) - P(\overrightarrow{S_{i} \: S_{j}} | \mathbf{S}_H) \right) ^2 \] \end{define}
This definition is very general and is meant to capture all the different ways that probability distributions could change due to different types of historical sets (see Page \cite{page06} for several types of path dependence).  This therefore admits to a refinement for each type of historical set dependence: exact path leading to $S_i$, unordered collection of states preceding $S_i$, existence of a path $\widetilde{S_{g} \: S_{i}}$ for different $S_{g}$, stability of $\mathbf{S}_H$, length of $\mathbf{S}_H$, and many others.  But insofar as each of these conditions can alter the transition probability distribution, their degree of path dependence can be measure in the same way.  

\begin{figure}[!ht]
\centering \caption{Path Dependence of $S_{25}$}
\label{PathDependenceFigure}
\begin{center}
\includegraphics{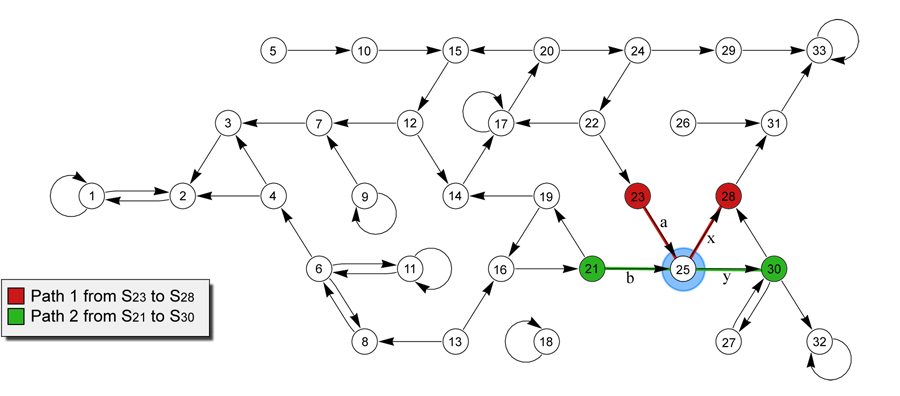}
\end{center}
\end{figure}
\begin{exa}Consider the situation depicted in figure \ref{PathDependenceFigure}; both states $S_{23}$ and $S_{21}$ transition into state $S_{25}$ which may transition to either $S_{28}$ or $S_{30}$ (edges $x$ and $y$ respectively).  The path dependence of $S_{25}$'s transitions is revealed through the Markov model when there are significant correlations in (say) $\overrightarrow{S_{23} \: S_{25} \: S_{28}}$ and $\overrightarrow{S_{21} \: S_{25} \: S_{30}}$.  So even though it may be the case that $P(s_{t+1} = S_{28} | s_{t} = S_{25}) = P(s_{t+1} = S_{30} | s_{t} = S_{21})$, i.e. $P(x) = P(y)$, it may not be the case that $P(x|a) = P(y|a)$ or that $P(x|b) = P(y|b)$.  Assume that the system dynamics enter $S_{25}$ equally often from both $S_{23}$ and $S_{21}$.  Given the system at $S_{25}$ let's assume $P(x) = P(y) = 0.5$.  Analyzing the individual time series of data may reveal that $P(x|a) = 0.8$ and $P(y|b) = 1.0$.  So $S_{25}$'s path dependence on $S_{23} = (0.5 - 0.8)^2 + (0.5 - .2)^2 = 0.18$.  $S_{25}$'s path dependence on $S_{21} = (0.5 - 0.0)^2 + (0.5 - 1.0)^2 = 0.5$.  Considering the conditional probabilities $P(x|a) = 0.8$ and $P(y|b) = 1.0$ the relative value of these figures matches intuition. \end{exa}

\section{Future Work}
This paper is very much a work in progress and so each definition, measure, theorem, algorithm and example is a subject for future work.  In addition to bolstering the current contents I have several planned avenues for expansion and spin-off projects, some of which are outlined in this section.  

\subsection{Application to Data Sets and Existing Models}

Good methodology exists as a facilitator to good science, so the first and perhaps most important extension of this project is to apply these measures to models within substantive research projects.  Over the past year of development I have given a few presentations and have engaged in many conversations about these techniques and their potential to reveal interesting features of system dynamics.  The potential collaborative projects that resulted from these discussions (outlined below) are in addition to planned applications to personal research in the evolution of culture and morality, institutional design, biological contagion control, supply chain management, ecological robustness, resource sustainability, and various philosophical implications of seeing multiple systems' dynamics as instances of the same underlying phenomena. 

Within the University of Michigan several parties have expressed interest in this methodology.  Qing Tian (working with Dan Brown) plans to use vulnerability analysis to identify the social and physical levers of the well-being of people in the Poyang Lake area of China where flooding frequently disrupts economic and social activity.  Abe Gong in the department of political science and public policy wants to incorporate these tools for research into far-from-equilibrium dynamics in organizational change.  Dominick' Wright (working with Scott Atran) want to find points of susceptibility to terrorist cell activity; intelligently disrupting sustainable operation of terrorist networks might offer low-impact methods to benefit national security.  Warren Whatley of the Department of Economics and the Center for African America Studies would like to apply the path dependency measures to his data and econometric model of the 18$^{\mathrm{th}}$ century British slave trade in America to identify any lasting effects on African economic and political stability.  Chris Chapman of the Multimedia Development Team at the University of Michigan Medical School wants to uncover user behavior patterns in educational technology in order to improve retention and pinpoint weak links to improve educational efficacy.  Several other individuals have expressed interest in the technique, but the above selection suffices to demonstrate the potential benefits this methodology has for substantive research fields across multiple disciplines.

This methodology has also garnered interest from the private sector.  State Farm insurance has expressed interest in building a model to better understand the effects of word-of-mouth spread in insurance company choices that exploits tipping behavior and path dependence.  Lockheed Martin is pushing to develop software to read in data from existing simulations and engineering tools to build the Markov model and identify the features of system dynamics defined herein.  Palantir Technologies is considering incorporating these analysis capabilities into their social network and financial market analysis software platforms.  Because the methods are designed to be as general possible, algorithms implementing them could be constructed as code libraries for popular scientific programming languages (Java, C, C++, Python, Matlab, Mathematica, etc.) and made open source for popular consumption and widespread use.  I have started to pursue government grant funding to develop this software package.

Through each of these applications and others that follow I expectt to uncover exceptions, caveats, refinements, and alternatives to the work as it stands now.  I also expect to discover many more features of system dynamics that can be uncovered via the Markov model representation.\footnote{In fact, I already have many more measures not presented in the current work because they do not fit snugly into one of the three general categories of dynamics presented.}  This paper presents a first attempt at capturing the above-defined properties; through applied research I will be able to validate these measures' usefulness and improve them where necessary.

\subsection{Technique To Construct Markov Diagram from Models or Data}

As mentioned in the text, I have sketched a methodology to generate the Markov model representation from data sets (whether from a database or collected from a generative model).  Because actual data sets and parameter spaces are typically quite large, a tool to automatically create the Markov transition diagram is necessary to perform the above measurements.  Part of designing this tool will be specifying precisely what the restrictions and requirements of a data set are.  Another part includes providing the algorithms to use available data to generate a system's states and transitions properly weighted.  Because the methodology presented in this paper is essentially a statistical technique (more on this below) some time will be spent to demonstrate that the assumptions made are the minimal assumptions and that the structure generated is the most justified result from the observable sample.

\subsection{Equivalence Classes for System Dynamics}

One major goal of complex systems research is to identify common underlying mathematical properties in a myriad of seemingly very different phenomena.  The Markov modeling technique allows us to create a common representation of almost any system's dynamics.  Differences in the definitions of system states, however, will still mask many of the similarities.  That difficulty notwithstanding we can make great gains by identifying network \textit{motifs} (repeated patterns in the graph structure) and establishing cross-disciplinary equivalence classes of system behavior.  Achieving this goal will require solving issues with the choice of system resolution and ``playing with'' the resolution to find the matching patterns.  Though this may sound suspicious, changing the resolution is nothing more than altering the level of organization to which we are applying the properties.  As long as we are consistent in our application of these techniques then we may be able to discover similarities in many complex systems' dynamics.

\subsection{Non-Probabilistic Definitions of These Phenomena}

There are two potential non-trivial objections to the above-given probabilistic accounts of properties of system dynamics.  The first is that probabilistic definitions are inadequate because we aim to understand these features as properties that systems possess rather than dynamics they \textit{might} have.  To answer this question we may first need a better grasp of dispositional versus categorical properties more generally (see below).  But it may be the case the other definitions cashed out purely in terms of structural properties of Markov models may be what some people would find more intuitive.  It could also be that the definitions these potential objectors are seeking cannot be formulated within Markov models at all.  As long as the above definitions reveal useful distinctions and patterns of system behavior the project was a success, but still better (or at least different and also useful) measures may be available if build from a different formal foundation.  I will, naturally, continue to pursue other and hugely different measures of system dynamics.
    
The other objection to the probabilistic definitions provided is that a person may insist that for many of these concepts the definition is incomplete without the causal explanation for how it comes about.  Like all other statistics-like approaches (see ``metastatistics'' below) these measures may be realized by many different micro-level dynamics.  Some of those dynamics may not seem proper candidates for robustness or tipping behavior even if the data they generate reveals it as such from this analysis.  But if this were to happen then I would consider the project a huge success.  This would be similar to discovering scale-free degree distributions in many different networks from disparate research fields.  Finding that common property urged researchers to pursue more deeply the phenomena and they eventually uncovered several different mechanisms by which scale-free network may be created.  Our understanding of each of those systems greatly increased because we had a common yardstick with which to measure them.  The probabilistic measures presented here are not intended to replace or make unnecessary the deeper scientific analysis - they are supposed to foster it.

\subsection{Investigating Dispositional Properties}

While working on the formal definitions of robustness-related properties I realized that these are all dispositional properties; and dispositional properties constitute a long-standing philosophical problem. That connection immediately made me wonder if my mathematical formalism might shed some new light on how to differentiate dispositional properties from categorical ones.  Dispositional properties are philosophically troubling for many reasons, but primarily these stem from the question of whether their subjunctive conditional status distinguishes them from categorical properties. It may therefore be interesting to consider how to apply my methodology to investigate other dispositional properties such as soluble, malleable, affordable, and differentiate them from other properties with accepted necessary and sufficient conditions. If in addition to the obvious and direct application of this methodology to improve the performance capabilities of systems my research also produces insights into the nature of dispositional properties in general that would be unanticipated but certainly welcome news. Let's look at some details of the problem.

A big part of the problem is that any property can be given a subjunctive conditional description, but a property would be dispositional if and only if such a definition is the only possible one. Color properties are known as primary properties and should be excellent candidates of clear-cut categorical properties. Yet being red is dispositional in the sense that nothing seems red in the dark.  Whether or not redness (or conductive, or triangular, or $\ldots$) is dispositional in the same way that fragility is has not been solved. My mathematical analysis reveals how particular dispositional properties (robustness-related ones) are behavioral in nature and emerge from the microbehavior of the system components. What I hope is that this can be expanded to develop necessary and sufficient conditions for properties to be dispositional in different ways $\ldots$ or at least a step in a helpful direction.

\subsection{Metastatistics}

I have said elsewhere in this paper that I consider the methodology presented here to be similar in kind to statistics.  It starts with data (perhaps generated from simulations of a model), fits a model (a Markov model) to the data, and then purports to describe the real system with measures over that model (my definitions).  Statistics as we usually see it takes a different kind of model (some form of distribution or estimator), but its purpose and general method of attack are very similar.  And this procedure is clearly different from other sorts of models in that neither standard statistics nor my methodology can explain the phenomena being analyzed.  Statistics (and my methodology) can produce evidence that some generative theory-driven model does explain the observed data, but the theory behind the generative model is what is doing the explaining.  
    
Standard statistics and my methodology are certainly not unique in their abilities to measure but not explain phenomena.  Much of complex network analysis can be seen in this light as well.  The network representation facilitates the calculation of measures on the generating data but not because the links identified in the network representation are in the actual system's features.  Classifier systems, Bayes nets, hidden Markov models, and neural nets are all further examples where the formal representation can permit measures and produce predictions without mirroring the structure and dynamics of the underlying behavior-generating system.  Seeing all these different techniques under the same metastatistical light may allow us to 1) bridge gaps among these techniques, 2) identify broader guidelines for the proper application and interpretation of these techniques, and 3) find new statistics-like techniques with desired features.

\end{document}